\documentclass{amsart}
\usepackage{amssymb}
\usepackage{etoolbox}
\usepackage{xcolor}
\usepackage{tikz}
\usetikzlibrary{arrows,calc,positioning}
\usepackage{hyperref}
\usepackage[capitalise]{cleveref}
\tikzset{
  >=stealth,
  node/.style={circle,minimum size=1.1mm,inner sep=0mm,draw=black,fill},
  hyperedge/.style={rectangle,minimum width=5mm,minimum height=5mm,draw,inner sep=0mm}
}

\theoremstyle{plain}
\newtheorem{theorem}{Theorem}
\newtheorem{lemma}{Lemma}
\newtheorem{corollary}{Corollary}
\newtheorem{proposition}{Proposition}
\theoremstyle{remark}
\newtheorem{example}{Example}
\newtheorem{remark}{Remark}
\theoremstyle{definition}
\newtheorem{definition}{Definition}

\newcommand{\dom}{\mathrm{dom}}
\newcommand{\eqdef}{\mathrel{\mathop:}=}
\newcommand{\Nat}{\mathbb{N}}
\newcommand{\parto}{\rightharpoonup}
\newcommand{\ran}{\mathrm{ran}}
\newcommand{\shuffle}{\mathbin{||}}
\newcommand{\Sym}{\mathrm{Sym}}

\newcommand{\DTIME}{\mathop{\mathrm{DTIME}}}
\newcommand{\EXPTIME}{\textup{EXPTIME}}
\newcommand{\NP}{\textup{NP}}
\newcommand{\PSPACE}{\textup{PSPACE}}

\newcommand{\ass}{\mathit{ass}}
\newcommand{\att}{\mathit{att}}
\newcommand{\ext}{\mathit{ext}}
\newcommand{\HRG}{\mathcal{HRG}}
\newcommand{\HRL}{\mathcal{HRL}}
\newcommand{\lab}{\mathit{lab}}
\newcommand{\Language}{\mathit{L}}
\newcommand{\parikh}{\Psi}
\newcommand{\Patch}{\mathfrak{P}}
\DeclareMathOperator{\red}{red}
\newcommand{\SG}{\mathrm{sg}}
\newcommand{\type}{\mathit{type}}
\newcommand{\unl}{\mathrm{unl}}

\newcommand{\UniClassHRG}{\mathcal{C}}
\newcommand{\UniProp}{\mathit{PROP}}
\newcommand{\UniSelA}{\sigma}
\newcommand{\UniSelB}{\theta}
\newcommand{\UniSetSelsB}{\Theta}

\title[Complexity of Decision Problems for HRGs]{Complexity of Decision Problems \\ for Hyperedge Replacement Grammars}

\author{Tikhon Pshenitsyn}

\address{Steklov Mathematical Institute of RAS, 8 Gubkina st., Moscow, Russia 119333}
\email{tpshenitsyn@mi-ras.ru}
\address{Ivannikov Institute for System Programming of RAS, 25 Alexander Solzhenitsyn st., Moscow, Russia 109004}

\begin{document}
	
\begin{abstract}
	This paper investigates complexity of various decision problems for hyperedge replacement grammars, starting with the uniform membership problem. It turns out that the complexity of the latter problem depends on how one defines a hypergraph. There are two commonly used definitions in the field, which differ in whether repetitions of attachment nodes of a hyperedge are allowed in a hypergraph or not. We show that, for arbitrary grammars, the problem under consideration is EXPTIME-complete, even for string-generating hyperedge replacement grammars, but it is NP-complete if repetitions are not allowed. 
	
	We extend the developed proof techniques in order to prove a general EXPTIME-hardness result that covers a lot of decision problems concerning hyperedge replacement grammars. In particular, checking whether a grammar generates some graph satisfying a property that cannot be recovered from its multiset of hyperedge labels is EXPTIME-hard. We call such properties non-Parikh. They include such well-known graph-theoretic properties as e.g.~connectedness, Eulerianity, Hamiltonianity, acyclicity. A tight upper bound is established for EXPTIME-compatible properties via the Filter Theorem.
\end{abstract}

\maketitle

\section{Introduction}\label{section:introduction}

A \emph{mildly context-sensitive grammar formalism} is a class of grammars that 
\begin{itemize}
	\item generate all context-free languages and also the language $\{ww \mid w \in \Sigma^\ast\}$;
	\item generate only languages from P (i.e.~that can be parsed in polynomial time) satisfying the constant growth property \cite{Kallmeyer10}.
\end{itemize}
These requirements are supposed to provide an optimal trade-off between expressivity and parsing complexity for modelling natural languages. Many mildly context-sensitive grammar formalisms have been considered in the literature, and they can be divided into equivalence classes according to what class of languages they generate. One equivalence class contains tree-adjoining grammars (TAGs), combinatory categorial grammars (CCGs), and linear indexed grammars. Another one, which generates a strictly wider class of languages, includes linear context-free rewriting systems (LCFRSs), multiple context-free grammars (MCFGs), and deterministic tree-walking transducers (DTWTs). Let us call the class of languages generated by any of the latter formalisms \emph{multiple context-free languages}.

The definition of a mildly context-sensitive grammar formalism characterises a class of languages rather than a class of grammars. However, as noted in \cite[p.~50]{Kallmeyer10}, ``in real natural language applications, we often deal with very large grammars [\dots] therefore, for natural language processing, the complexity of the universal recognition problem is an important factor.'' Besides, in applications, a grammar is subject to regular updating, and it is desirable to know how adding new rules to a grammar affects parsing complexity. Finally, one would like to somehow compare LCFRSs, MCFGs, and DTWTs in terms of their complexity, given that they all generate the same class of languages. All this motivates studying the universal recognition problem, which, following \cite{BjorklundBE16}, we call the \emph{uniform membership problem}: 

\begin{quote}
	\textsc{Instance.} A grammar $G$ and a word $w$.
	\\
	\textsc{Question.} Is $w$ generated by $G$?
\end{quote}

This problem has been studied for various mildly context-sensitive formalisms in several papers; we provide some of the known results in Table \ref{table:uniform-membership-mcsgf}.
\begin{table}[h!]
	\centering
	\begin{tabular}{|l|l|}
		\hline
		TAG & P \\
		\hline 
		CCG & \EXPTIME{}-complete~\cite{KuhlmannSJ18} \\
		\hline\hline
		LCFRS & \PSPACE{}-complete \cite{KajiNSK92} \\
		\hline
		MCFG & \EXPTIME{}-complete \cite{KajiNSK92} \\
		\hline
		DTWT & \EXPTIME{}-complete \cite{BjorklundBE16} \\
		\hline
	\end{tabular}
	\caption{Uniform membership problem for mildly context-sensitive grammar formalisms.}
	\label{table:uniform-membership-mcsgf}
\end{table} 

The original motivation of this work was to study the uniform membership problem for hyperedge replacement grammars, whose complexity, to the best of my knowledge, has not been previously established in the literature.  Hyperedge replacement grammars are a natural and well studied generalisation of context-free grammars to hypergraphs \cite{DrewesKH97, Habel92, Engelfriet97}. It is known that HRGs generate only languages from \NP{}, including an \NP{}-complete hypergraph language \cite{DrewesKH97}. One can consider HRGs that generate only string graphs, i.e.~graphs of the form $\vcenter{\hbox{{\tikz[baseline=.1ex]{
				\foreach \i in {1,...,4}
				{
					\node[node, label=\ifnumequal{\i}{1}{left}{right}:{\tiny \ifnumequal{\i}{1}{$(1)$}{\ifnumequal{\i}{4}{$(2)$}{}}}] (V\i) at ($(0.8*\i-0.8,0)$) {};
				}
				\node at ($(0.8*1.5,0)$) {\dots};
				\draw[-latex, thick] (V1) -- node[above] {$a_1$} (V2);
				\draw[-latex, thick] (V3) -- node[above] {$a_n$} (V4);
}}}}$ (such a graph corresponds to the string $a_1 \ldots a_n$). If an HRG generates only string graphs, then one can associate a string language with it. It is known that string-generating HRGs generate exactly multiple context-free languages \cite{EngelfrietH91} so HRG is a mildly context-sensitive grammar formalism equivalent to LCFRS, MCFG, and DTWT. Thus, it is interesting to explore the uniform membership problem both for string-generating HRGs and for HRGs in general.

It turns out that the complexity of the uniform membership problem for HRGs depends on the definition of a hypergraph used. If one considers repetition-free hypergraphs as is done in \cite{DrewesKH97} (i.e.~hypergraphs whose external function and all attachment functions are injective), then the uniform membership problem for HRGs is \NP{}-complete (\cref{theorem:upper_bound_NP,theorem:lower_bound_NP}). However, if one allows repetitions of attachment nodes and external nodes in hypergraphs, then the uniform membership problem becomes \EXPTIME{}-complete, even if input grammars are restricted to be string-generating (\cref{theorem:upper_bound_EXPTIME,theorem:lower_bound_EXPTIME}). 

Proofwise, the most interesting result is \EXPTIME{}-hardness of the uniform membership problem for string-generating HRGs (\cref{theorem:lower_bound_EXPTIME}). One might think that it is proved by simply constructing a polynomial reduction from the uniform membership problem for MCFGs or for DTWTs. However, this is not the case. Indeed, consider Table \ref{table:uniform!}, which summarises the complexity of the uniform membership problem for various formalisms over a one-letter alphabet.
\begin{table}[h!]
	\centering
	\begin{tabular}{|c|c|c|}
		\hline
		& $\vert T \vert = 1$ & $\vert T \vert > 1$
		\\\hline
		LCFRS & P & \PSPACE{}-complete
		\\\hline
		MCFG & \multicolumn{2}{c|}{\EXPTIME{}-complete}
		\\\hline
		DTWT & \multicolumn{2}{c|}{\EXPTIME{}-complete}
		\\\hline
		string-generating & P & \EXPTIME{}-complete
		\\
		HRG & (\cref{remark:one-letter-sg-HRG}) & (\cref{theorem:lower_bound_EXPTIME})
		\\\hline
	\end{tabular}
	\caption{Uniform membership problem for grammars generating multiple context-free languages depending on the size of a terminal alphabet $T$. Note that, for arbitrary HRGs over a one-letter alphabet, the problem is \EXPTIME{}-complete (see \cref{remark:one-letter-HRG}).}
	\label{table:uniform!}
\end{table}
If there were a polynomial-time procedure that transforms an MCFG (a DTWT) into an equivalent string-generating HRG, then, over a one-letter terminal alphabet, it would give a polynomial reduction of an \EXPTIME{}-complete problem to a problem from P, which is a contradiction. Instead, the problem used for a reduction is the following one, shown to be \EXPTIME{}-hard in \cref{proposition:EXPTIME-LCFRS}:
\begin{quote}
	\textsc{Instance.} An LCFRS $G$ over the two-letter alphabet $\{a,b\}$.
	\\
	\textsc{Question.} Is the intersection $\Language(G) \cap b\{a,b\}^\ast$ non-empty?
\end{quote}

The argument developed to prove \EXPTIME{}-hardness of the uniform membership problem for string-generating HRGs turns out to have a great potential for generalization. We reuse it to prove \EXPTIME{}-hardness of a wide range of decision problems regarding HRGs (\cref{theorem:main-meta}). In particular, we show that checking whether a given HRG generates some\slash only string\slash connected\slash Eulerian\slash Hamiltonian\slash acyclic graphs is \EXPTIME{}-hard (\cref{corollary:non-Parikh}). In fact, this holds for any graph property $P$ such that there are two graphs of the same type, with at least one node each and with the same Parikh vector, one of which satisfies $P$ while the other does not. We call such properties \emph{non-Parikh}. For repetition-free HRGs, we obtain an analogous lower bound for properties invariant under \emph{patch} transformations (\cref{theorem:main-meta-repetition-free}).

In order to match this lower \EXPTIME{} bound, we develop a notion of \EXPTIME{}-compatible properties, a time-bounded version of the notion of compatible properties defined in \cite[Ch.~6, Sec.~2]{Habel92}, and show that deciding if an HRG generates at least one graph satisfying such a property is in \EXPTIME{} (\cref{theorem:bounded-filter-theorem}). In particular, this result and the previous one imply that checking whether an HRG is string-generating is \EXPTIME{}-complete (\cref{corollary:string-generating}), since the property of not being a string graph is both non-Parikh and \EXPTIME{}-compatible.

\subsection*{Related Work}

There is an article ``Uniform parsing for hyperedge replacement grammars'' \cite{BjorklundDES21}, which is concerned with finding a subclass of HRGs for which a \emph{polynomial-time} uniform parsing algorithm exists. The goal of the present paper is different: we aim to study the uniform parsing problem for all HRGs without expecting it to be polynomial, of course. In personal communication \cite{Drewes_personal}, Professor~Frank Drewes encouraged me to study complexity of the uniform membership problem for HRGs, and he conjectured that it was related to the membership problem for right-linear grammars over permutation groups. The proof of \cref{theorem:upper_bound_NP} indeed relies on this problem.

As for other decision problems for hyperedge replacement grammars, one predecessor work is \cite{Wanke94}. Its introduction says: \textit{A precise complexity analysis of simple graph properties on context-free sets of graphs is left open} \cite{Wanke94}. Then, it focuses on the following problem: given a graph grammar, check whether it generates some\slash only connected graphs. Graph grammars considered in that work are node label controlled (NLC) graph grammars, a kind of context-free graph grammars where nodes are replaced by graphs. Wanke proves \EXPTIME{}-completeness of this problem, establishing the lower bound for apex NLC grammars, a subclass that can also be defined using HRGs. Thus the \EXPTIME{}-hardness of these connectivity problems for HR systems was already established in \cite{Wanke94}. This will also follow from our general metatheorems (\cref{theorem:main-meta,theorem:main-meta-repetition-free}).

\subsection*{Structure of the Paper}

The paper is organised as follows. In \cref{section:preliminaries}, we define hypergraphs and hyperedge replacement grammars. In \cref{section:chain}, we show how transformation monoids and permutation groups can be modelled in HRGs using chain productions, which is used for proving upper bounds; this also provides some intuition behind the difference between repetition-free HRGs and general ones. In \cref{section:upper-bounds}, we prove upper bounds for the uniform membership problem. In \cref{section:lower-bounds}, we prove the tight lower bounds. In \cref{section:complexity-decision-problems}, we establish complexity results for various decision problems regarding HRGs, introducing, in particular, \emph{non-Parikh graph properties} and \emph{\EXPTIME{}-compatible graph properties}. In \cref{section:conclusion}, we conclude.

\section{Preliminaries}\label{section:preliminaries}

We denote the set $\{1,\dotsc,n\}$ by $[n]$. The following definitions are mainly taken from \cite{Engelfriet97} (with slight modifications) because they are the most general among those used in the classical textbooks \cite{Habel92,DrewesKH97,Engelfriet97} and the handiest to reason about.
\begin{definition}[$\Sigma$-typed alphabet]
	Given an alphabet $\Sigma$ of \emph{selectors}, a $\Sigma$-typed alphabet is a set $C$ along with a function $\type:C \to \mathcal{P}(\Sigma)$ such that $\type(c)$ is finite for $c \in C$.
\end{definition}

\begin{definition}[Hypergraph]
	Let $C$ be a finite $\Sigma$-typed alphabet of hyperedge labels. A \emph{hypergraph} over $C$ is a tuple $H = \langle V_H,E_H, \lab_H, \att_H, \ext_H\rangle$ where  $V_H$ is a finite set of nodes; $E_H$ is a finite set of hyperedges; for each $e \in E_H$, $\att_H(e) : \Sigma \parto V_H$ is a partial function with a finite domain; $\lab_H:E_H \to C$ is a labeling function such that $\type(\lab_H(e))=\dom(\att_H(e))$; $\ext_H : \Sigma \parto V_H$ is a partial function with a finite domain. Elements of $\ran(\ext_H)$ are called \emph{external nodes}. Let $\type(H) \eqdef \dom(\ext_H)$ and let $\type_H(e) \eqdef \dom(\att_H(e))$ for each $e \in E_H$. We also write $\type(e)$ for $\type_H(e)$ when $H$ is clear from the context.
	
	The set of hypergraphs over $C$ is denoted by $\mathcal{H}(C)$.
\end{definition}
In drawings of hypergraphs, nodes are depicted as black circles and hyperedges are depicted as labeled rectangles. When depicting a hypergraph $H$, we draw a line with a label $\UniSelA$ from $e$ to $v$ if $\att_H(e)(\UniSelA)=v$. External nodes are represented by numbers in round brackets: if $\ext_H(\UniSelA)=v$, then we mark $v$ as $(\UniSelA)$. 

\begin{example}
	The hypergraph $H$ depicted below has five nodes and five hyperedges. The $f$-labelled hyperedge $e$ has the type $\type_H(e) = \{1,2,3\}$. The type of the hypergraph $\type(H)$ equals $\{1\}$, as it has only one external node marked as $(1)$.
	\[
	\vcenter{\hbox{{\tikz[baseline=.1ex]{
					\node[node, label=right:{\tiny $(1)$}] (N1) {};
					\node[hyperedge, below=3.6mm of N1] (O) {$f$};
					\node[node, below left = 0.3mm and 5mm  of O] (N2) {};
					\node[node, below right = 0.3mm and 5mm  of O] (N3) {};
					\node[hyperedge, below=3.6mm of N3] (O3) {$g$};
					\node[node, below left = 0.3mm and 5mm  of O3] (N31) {};
					\node[node, below right = 0.3mm and 5mm  of O3] (N32) {};
					\node[hyperedge,below=3.6mm of N2] (E2) {$x$};
					\node[hyperedge,below=3.6mm of N31] (E31) {$y$};
					\node[hyperedge,below=3.6mm of N32] (E32) {$z$};
					\draw[-] (N1) -- node[left] {\tiny $1$} (O);
					\draw[-] (O) -- node[above] {\tiny 2} (N2);
					\draw[-] (O) -- node[above] {\tiny 3} (N3);
					\draw[-] (N3) -- node[left] {\tiny 1} (O3);
					\draw[-] (O3) -- node[above] {\tiny 2} (N31);
					\draw[-] (O3) -- node[above] {\tiny 3} (N32);
					\draw[-] (N31) -- node[left] {\tiny 1} (E31);
					\draw[-] (N32) -- node[right] {\tiny 1} (E32);
					\draw[-] (N2) -- node[left] {\tiny 1} (E2);
	}}}}
	\]
\end{example}

\begin{definition}[Graph]
	If $\type_H(e) = \{1,2\}$ for a hyperedge $e$, then this hyperedge is called an \emph{edge} and it is depicted by an arrow going from $\att_H(e)(1)$ to $\att_H(e)(2)$. A hypergraph $H$ is a \emph{graph} if each of its hyperedges is an edge.
\end{definition}

Below, we define a basic hypergraph called a \emph{handle} in \cite{DrewesKH97}.
\begin{definition}[Handle]
	Given $a \in C$, $a^\bullet$ is a hypergraph such that $V_{a^\bullet} = \type(a)$; $E_{a^\bullet} = \{e\}$ with $\lab_{a^\bullet}(e)=a$ and $\type_{a^\bullet}(e)=\type(a)$; $\att_{a^\bullet}(e)(\UniSelA)=\ext_{a^\bullet}(\UniSelA) = \UniSelA$ for $\UniSelA \in \type(a)$.
\end{definition}

We shall be particularly concerned with string graphs which are graphs with linear structure representing strings.

\begin{definition}[String graph]
	\emph{A string graph $\SG(w)$ induced by a string $w=a_1\dots a_n$} is defined as follows: $V_{\SG(w)} = \{v_0,\dotsc,v_n\}$, $E_{\SG(w)} = \{s_1,\dotsc,s_n\}$; $\type(s_i) = \type(\SG(w)) = \{1,2\}$, $\att_{\SG(w)}(s_i)(1)=v_{i-1}$, $\att_{\SG(w)}(s_i)(2)=v_{i}$, $\lab_{\SG(w)}(s_i)=a_i$ (for $i = 1, \ldots, n$); $\ext_{\SG(w)}(1)=v_0$, $\ext_{\SG(w)}(2)=v_n$. 
\end{definition}

Now, let us recall the standard quotient construction used to define hyperedge replacement.
\begin{definition}[Hypergraph quotient]
	Let $H$ be a hypergraph and let $R$ be a binary relation on $V_H$. Let $\equiv_R$ be the smallest equivalence relation on $V_H$ containing $R$. Then $H/R = H^\prime$ is the following hypergraph: $V_{H^\prime} = \{[v]_{\equiv_R} \mid v \in V_H\}$; $E_{H^\prime} = E_H$; $\lab_{H^\prime} = \lab_H$; $\att_{H^\prime}(e)(s) = [\att_{H}(e)(s)]_{\equiv_R}$; $\ext_{H^\prime}(s) = [\ext_{H}(s)]_{\equiv_R}$.
\end{definition}

\begin{definition}[Hyperedge replacement]
	Let $H,K$ be two hypergraphs over $C$; let $e \in E_H$ be a hyperedge such that $\type(e) = \type(K)$. Then the \emph{replacement of $e$ by $K$ in $H$} (the result being denoted by $H[e/K]$) is defined as follows:
	\begin{enumerate}
		\item Remove $e$ from $H$ and add a disjoint copy of $K$. Formally, let $L$ be the hypergraph such that $V_L = V_H \sqcup V_K$, $E_L = (E_H \setminus \{e\}) \sqcup E_K$, $\lab_L$ is the restriction of $\lab_H \cup \lab_K$ to $E_L$, $\att_L$ is the restriction of $\att_H \cup \att_K$ to $E_L$, and $\ext_L = \ext_H$.
		\item Glue the nodes that are incident to $e$ in $H$ with the external nodes of $K$. Namely, let $H[e/K] \eqdef L/R$ where $R = \{(\att_H(e)(s),\ext_K(s)) \mid s \in \type(e)\}$.
	\end{enumerate}
\end{definition}

\begin{definition}[Hyperedge replacement grammar]
	A \emph{hyperedge replacement grammar} is a tuple $\Gamma = \langle N, T, \Sigma, P, S\rangle$ where
	\begin{itemize}
		\item $N$ and $T$ are finite disjoint $\Sigma$-typed alphabets;
		\item $P$ is a finite set of \emph{productions} of the form $X \to D$ where $X \in N$ and $D$ is a hypergraph over $N \cup T$ such that $\type(X) = \type(D)$;
		\item $S \in N$ is the initial nonterminal label.
	\end{itemize}
	If $H$ is a hypergraph with $e \in E_H$ such that $\lab_H(e) = X$ and if $p = (X \to D) \in P$, then we say that $H[e/D]$ can be directly derived from $H$ and we write $H \Rightarrow H[e/D]$ (or $H \Rightarrow_{\Gamma} H[e/D]$, or $H \Rightarrow_p H[e/D]$ if we want to specify a grammar or a production used at this step). The language $\Language(\Gamma)$ generated by $\Gamma$ consists of hypergraphs $H \in \mathcal{H}(T)$ such that $S^\bullet \Rightarrow^\ast H$.
	
	$\Gamma$ is said to be of \emph{order $r$} if $\vert\type(A)\vert \le r$ for all $A \in N$ \cite{DrewesKH97}.
	
	The class of all HRGs is denoted by $\HRG$.
\end{definition}
\begin{remark}\label{remark:abstract-concrete}
	Following standard conventions, we do not distinguish carefully between hypergraphs and their isomorphism classes throughout the paper (see e.g.~\cite{DrewesKH97}). For example, we shall assume that $\Language(\Gamma)$ consists of classes of isomorphic hypergraphs rather than of concrete hypergraphs. Indeed, we do not care about the particular nodes and hyperedges in a generated hypergraph but only about its structure, i.e.~how nodes and hyperedges are connected to each other. However, when we write $H \in \Language(\Gamma)$, we treat $H$ as a concrete hypergraph, in particular, we allow ourselves to refer to its nodes or hyperedges.
\end{remark}

For some problems, it will be useful to erase the graph structure of a hyperedge replacement grammar and to work with a context-free grammar instead.

\begin{definition}[Linearisation]
	A \emph{linearisation of an HRG $\Gamma = \langle N, T, \Sigma, P, S\rangle$} is a context-free grammar $\langle N,T,P^\prime,S\rangle$ where $P^\prime$ consists of productions of the form $A \to \lab_H(e_1) \ldots \lab_H(e_k)$ for $(A \to H) \in P$, with $e_1,\ldots,e_k$ being the hyperedges of $H$ listed in arbitrary order. (A linearisation of an HRG is not unique.)
\end{definition}

As we mentioned in \cref{section:introduction}, different textbooks and articles in the field of hyperedge replacement grammars use different definitions of a hypergraph. In particular, it is common to require that a hypergraph is repetition-free:
\begin{definition}[Repetition-free hypergraph]
	A hypergraph $H$ is \emph{repetition-free} if the function $\ext_H$ is injective and $\att_H(e)$ is injective for every $e \in E_H$.
\end{definition}
In \cite{DrewesKH97}, for example, hypergraphs are repetition-free by definition.

The following notion is related to the above one, although it does not rule out repetitions in attachment functions of hyperedges completely.
\begin{definition}\label{definition:repetition-free}
	An HRG is \emph{repetition-free} if, for any of its productions $X \to D$, $\ext_D$ is injective and $\att_D(e)$ is injective whenever $\lab_D(e)$ is nonterminal.
\end{definition}
Sometimes, we shall call arbitrary HRGs \emph{repetition-allowing} to emphasize the difference between those and repetition-free HRGs. Note that repetition-freeness of a grammar, as defined here, permits repetitions of attachment nodes of terminal hyperedges.

The following is proved in \cite[Theorem 3.15]{Engelfriet97}:
\begin{quote}
	For every HRG $\Gamma = \langle N, T, \Sigma, P, S\rangle$ such that $\type(S) = \emptyset$, one can construct a repetition-free HRG $\Gamma^\prime$ such that 
	$\Language(\Gamma^\prime) = \Language(\Gamma)$.
\end{quote}
The method of constructing $\Gamma^\prime$ from $\Gamma$ presented in \cite{Engelfriet97} is not polynomial-time; we shall analyse it in \cref{section:upper-bounds}.

\section{On Empty and Chain Productions in HRGs}\label{section:chain}

To give the reader an initial feeling of the difference between repetition-free grammars and repetition-allowing ones, let us define the standard notions of empty and chain productions for HRGs \cite{Habel92}.
\begin{definition}
	A production $X \to D$ is \emph{empty} if $E_D = \emptyset$ and $V_D = \ran(\ext_D)$.
\end{definition}
\begin{definition}\label{definition_chain}
	A production $X \to D$ is \emph{chain} if $E_D = \{e\}$, $\lab_D(e)$ is nonterminal, and $V_D = \ran(\ext_D)$.
\end{definition}
In the case of ordinary context-free grammars generating words, there is only one empty production with $X$ in the left-hand side (which is $X \to \varepsilon$) and only one chain production with $X$ in the left-hand side and $Y$ in the right-hand side (which is $X \to Y$). This allows one to find in polynomial time all pairs $(A,B)$ such that there exists a derivation of the form $A=A_1 \Rightarrow \dotsc \Rightarrow A_k \Rightarrow B$ where $A_i \in N$ and either $B \in N$ or $B = \varepsilon$. Consequently, one can effectively eliminate empty and chain productions and obtain a context-free grammar where each rule strictly increases the size of a derivable string. If one could efficiently eliminate empty and chain productions in HRGs, this would yield an \NP{} algorithm for checking membership immediately. However, eliminating empty and chain productions is a costly procedure in the hypergraph case.
\begin{definition}\label{definition_production-function}
	Given a label $B$ and a function $f:\type(B)\to[n]$, define $F=F(B,f)$ by $V_F=[n]$, $E_F=\{e_f\}$, $\lab_F(e_f)=B$, $\att_F(e_f)(j)=f(j)$ for $j\in\type(B)$, and $\ext_F(j)=j$ for $j\in[n]$. If $\type(S)=[n]$ and $f:[n]\to[n]$, let $p(f)$ be the production $S\to F(S,f)$.
\end{definition}
The production $p(f)$ is a chain production. A simple but important observation is that the composition of rules $p(f)$ and $p(g)$ corresponds to the composition of functions $f g=f\circ g$.
\begin{lemma}\label{lemma_composition}
	Let $S^\bullet \Rightarrow_{p(f)} L \Rightarrow_{p(g)} K$. Then $K = F(S,f g)$.
\end{lemma}
\begin{proof}
	Clearly, $L = F(S,f)$ and $K = F(S,f)[e_f/F(S,g)]$. The latter hypergraph is obtained from the former by removing $e_f$, adding $e_g$ and identifying nodes by means of $\equiv_R$ where $R = \{(\att_{L}(e_f)(s),\ext_{F(S,g)}(s)) \mid s=1,\dotsc,n\}$. Consequently, 
	\begin{multline*}
		\att_K(e_g)(s) = [\att_{F(S,g)}(e_g)(s)]_{\equiv_R} = [\ext_{F(S,g)}(g(s))]_{\equiv_R} = [\att_{L}(e_f)(g(s))]_{\equiv_R} \\ =
		[\ext_{L}(f(g(s)))]_{\equiv_R} = \ext_K(f(g(s))). \qedhere
	\end{multline*}
\end{proof}
Therefore, one can simulate transformation monoids over finite sets using chain productions. To recall, a transformation monoid consists of functions $f: X \to X$ (called transformations of $X$) for some fixed $X$, with the monoid operation being composition. Given any set $\mathcal{F}$ of transformations of $X$, let $\langle \mathcal{F} \rangle$ denote the least submonoid containing $\mathcal F$.

\begin{example}\label{example_all_permutations}
	Let $f,g: [n] \to [n]$ be some generators of the permutation group $\Sym(n)$. Consider the grammar $G = \langle\{S\},\{a\},\Sigma, P,S\rangle$ where $\Sigma = \type(S)=\type(a)=[n]$ and $P = \{p(f),p(g),S \to a^\bullet\}$. 
	The grammar $G$ has the size polynomial in $n$. The language of this grammar consists of hypergraphs of the form $F(a,h_{1}\dotsc h_k)$ where $h_i \in \{f,g\}$ (this follows from \cref{lemma_composition}). Thus, $\Language(G) = \{F(a,h) \mid h \in \Sym(n)\}$, hence $\vert \Language(G) \vert = n!$. This implies that eliminating empty and chain productions in $G$ necessarily leads to exponential growth of the number of productions.
\end{example}

Let us consider the following HRG.
\begin{definition}
	Let $\mathcal{F} \subseteq \{f: [n] \to [n]\}$ be a finite set of transformations of $[n]$. Then $G(\mathcal{F}) = \langle\{S\},\{a\},\Sigma, P(\mathcal{F}),S\rangle$ where $\Sigma = \type(S)=\type(a)=[n]$ and $P(\mathcal{F}) = \{p(f) \mid f \in \mathcal{F} \} \cup \{S \to a^\bullet\}$.
\end{definition}
\begin{proposition}\label{proposition:submonoid}
	$\Language(G(\mathcal{F})) = \{F(a,h) \mid h \in \langle \mathcal{F} \rangle\}$. 
\end{proposition}
This follows directly from \cref{lemma_composition}.
As a consequence, checking whether the hypergraph $F(a,h)$ belongs to $\Language(G(\mathcal{F}))$ is equivalent to checking whether $h \in \langle\mathcal{F} \rangle$. The latter problem is \PSPACE{}-complete \cite{Kozen77} and, consequently, the uniform membership problem for HRGs is \PSPACE{}-hard. Since we are going to show \EXPTIME{}-hardness of this problem, this lower bound is not yet enough.

Now, let us turn to repetition-free HRGs. Note that $G(\mathcal{F})$ is repetition-free if and only if $\mathcal{F}$ consists of bijections. This means that $\langle \mathcal{F} \rangle$ is a subgroup of the symmetric group $\Sym(n)$. The problem whether $h \in \langle \mathcal{F} \rangle$ is well known to be solvable in polynomial time using the Schreier-Sims algorithm \cite{Sims71,FurstHL80}. However, $G(\mathcal{F})$ uses only one nonterminal symbol $S$; clearly, one could use many nonterminals. This is why we need to consider the following problem.

\begin{quote}
	\textbf{The rational subset membership problem for symmetric groups \textsc{RatSym}}
	\\
	\textsc{Instance.} the set $\mathcal{F} \subseteq \Sym(n)$; a nondeterministic finite automaton $\mathcal{A} = \langle Q , \Delta, I, F \rangle$ over the alphabet $\mathcal{F}$; a permutation $\UniSelA \in \Sym(n)$.
	\\
	\textsc{Question.} is there a word $w = \UniSelA_1\dotsc \UniSelA_k \in \mathcal{F}^\ast$ accepted by $\mathcal{A}$ such that $\UniSelA_1\cdot \dotsc \cdot \UniSelA_k = \UniSelA$?
\end{quote}
It is known that \textsc{RatSym} is \NP{}-complete \cite{Khashaev22, LohreyRZ22}. We shall use this result in proving the \NP{} upper bound for repetition-free HRGs, since it helps to handle chain productions.

\section{Upper Bounds}\label{section:upper-bounds}

We start with proving upper bounds on the complexity of the uniform membership problem:

\begin{theorem}\label{theorem:upper_bound_EXPTIME}
	The uniform membership problem for HRGs is in \EXPTIME{}.
\end{theorem}

\begin{theorem}\label{theorem:upper_bound_NP}
	The uniform membership problem for repetition-free HRGs is in \NP{}.
\end{theorem}

The proof of \cref{theorem:upper_bound_EXPTIME} proceeds by eliminating empty and chain productions in an HRG and analysing the complexity of this procedure \cite{Habel92,Engelfriet97}. Despite this being a fairly standard exercise, let us recall the proof anyway to make the paper self-contained.

\begin{proof}[Proof of \cref{theorem:upper_bound_EXPTIME}]
	Let $\Gamma=\langle N,T,\Sigma,P,S\rangle$ be an input HRG and let $H_0$ be an input hypergraph. If $\type(H_0)\ne\type(S)$ or some hyperedge label of $H_0$ does not belong to $T$, we reject the input. Without loss of generality, let us assume that $\type(S)=\emptyset$ and that $H_0$ contains at least one hyperedge. Indeed, put $X=\type(S)$ and take a fresh terminal label $\$$ of type $X$. Define $G_S$ by $V_{G_S}=X$, $E_{G_S}=\{e_1,e_2\}$, $\lab_{G_S}(e_1)=S$, $\lab_{G_S}(e_2)=\$$, and $\att_{G_S}(e_i)(s)=s$ for $i=1,2$ and $s\in X$, with $\ext_{G_S}$ empty. Introduce a new start label $S^\prime$ of type $\emptyset$ and the production $S^\prime\to G_S$. Correspondingly, replace $H_0$ by the hypergraph obtained by adding a $\$$-labeled hyperedge whose attachment function is $\ext_{H_0}$, and then making the external function empty. The unique $\$$-labeled hyperedge records the original external function, so this transformation preserves membership.

	Under the above assumptions, let us perform a series of transformations.
	\begin{enumerate}
		\item\label{step-atmosttwo} We construct a grammar $\Gamma_1=\langle N_1,T,\Sigma_1,P_1,S_1\rangle$ equivalent to $\Gamma$ such that the start symbol $S_1$ does not occur in right-hand sides of productions and such that each right-hand side has at most two hyperedges \cite[Proposition 3.13]{Engelfriet97}. This is done in the same manner as for context-free grammars. Namely, if a right-hand side has more than two hyperedges, choose two of them, $e_1,e_2$, and list their distinct attachment nodes as $v_1,\ldots,v_k$. Replace $e_1,e_2$ by a hyperedge $f$ with a fresh nonterminal label $B$, where $\type(B)=[k]$ and $\att(f)(i)=v_i$. All nodes of the right-hand side are retained. Add a production $B\to D$ where $D$ consists of $v_1,\ldots,v_k$ and the original hyperedges $e_1,e_2$, with $\ext_D(i)=v_i$. Repeat this procedure until each right-hand side has at most two hyperedges. This construction preserves repetition-freeness if $\Gamma$ is repetition-free. The whole procedure is done in polynomial time, and hence $|\Gamma_1|\le p_1(|\Gamma|)$ for some polynomial $p_1$. Moreover, $\Gamma_1$ has order at most $R=R(|\Gamma|)$ for some polynomial $R$, since the number of selectors of a new nonterminal is bounded by the total number of nodes in an original right-hand side.

		\item\label{step-repetition-free} We construct a repetition-free HRG $\Gamma_2$ equivalent to $\Gamma_1$ using the construction from \cite[Theorem 3.15]{Engelfriet97}. In this construction, a nonterminal $A$ is replaced by copies indexed by partitions of $\type(A)$: a partition records which external nodes are identified in a derivation from $A$. Since $|\type(A)|\le R$, there are at most $(R+1)^R$ such partitions. For each production, the construction considers a partition for its left-hand-side nonterminal and a partition for each nonterminal hyperedge in its right-hand side, and performs the corresponding node identifications. There are at most two such hyperedges, so at most $(R+1)^{3R}$ choices are considered per production. Each resulting right-hand side has polynomial size in $|\Gamma_1|$. Thus $|\Gamma_2|\le 2^{p_2(|\Gamma|)}$ for some polynomial $p_2$, and the construction takes exponential time. It preserves the property of having at most two hyperedges in each right-hand side, and the order of $\Gamma_2$ does not exceed $R$.

		\item\label{step-elimempty} We eliminate empty productions in $\Gamma_2$ using the construction from \cite[Theorem 1.5]{Habel92}. Let $\Gamma_3$ be the resulting grammar. For each nonterminal $A$, there is only one possible empty production $A \to D$ with a repetition-free right-hand side, up to isomorphism: $D$ has no hyperedges and its nodes are precisely its distinct external nodes. Therefore, empty productions can be eliminated in the same way as for context-free grammars. First, compute the nonterminals deriving such an empty hypergraph. Then, for each production, omit any subset of the hyperedges labeled by these nonterminals, retaining all nodes, and discard the resulting empty productions. Since there are at most two hyperedges in each right-hand side, there are at most four choices per production. Consequently, $|\Gamma_3|\le p_3(|\Gamma_2|)$ for some polynomial $p_3$, and the construction is polynomial-time in $|\Gamma_2|$. It preserves repetition-freeness and does not increase the order. It also preserves the language, since every terminal hypergraph contains the $\$$-labeled hyperedge.

		\item We eliminate chain productions in $\Gamma_3$ using the construction from \cite[Theorem 1.8]{Habel92}. Let $\Gamma_3=\langle N_3,T,\Sigma_3,P_3,S_3\rangle$. The construction computes, in a deterministic way, the set
		\[
		\mathit{CHAIN}=\{A\to H\mid (A\to H)\text{ is chain},\;
		A^\bullet\Rightarrow^\ast_{\Gamma_3}H\}.
		\]
		Let $\mathit{or}_3$ be the order of $\Gamma_3$. A repetition-free chain production with left-hand side $A$ and sole hyperedge labeled by $B$ is determined, up to isomorphism, by the injective function from $\type(B)$ to $\type(A)$ specifying its attachments. In particular,
		\(
		|\mathit{CHAIN}|
		\le |N_3|^2(\mathit{or}_3+1)^{2\mathit{or}_3}.
		\)
		Since $\mathit{or}_3\le R$, it follows that
		\(
		|\mathit{CHAIN}|
		\le |\Gamma_3|^2\,2^{2R\log_2(R+1)}.
		\)
		
		Finally, we construct $\Gamma^\prime$ where a production is either a non-chain one from $P_3$ or a composition of a production from $\mathit{CHAIN}$ with a non-chain one from $P_3$. Thus the total number of productions in $\Gamma^\prime$ does not exceed
		\[
		2|\Gamma_3|^3\,2^{2R\log_2(R+1)}.
		\]
		Computing $\mathit{CHAIN}$ and forming these compositions takes exponential time in $|\Gamma|$.
	\end{enumerate}
	Summing up, we obtain an equivalent repetition-free grammar $\Gamma^\prime$ with no empty or chain productions, such that $|\Gamma^\prime|\le2^{p(|\Gamma|)}$ for some polynomial $p$. The whole procedure is done in exponential time.

	Now, let $H_0\in\Language(\Gamma^\prime)$. To bound the length of a derivation of $H_0$, define the measure function $\mu(H)=2|V_H|+|E_H|$. Applying a production $A\to D$ increases this measure by
	\[
	2(|V_D|-|\type(A)|)+|E_D|-1,
	\]
	since the grammar is repetition-free and no existing nodes are glued. If $D$ has no hyperedges, then it has at least one internal node, because empty productions have been eliminated. If $D$ has at least two hyperedges, the above expression is also positive. Finally, if $D$ has exactly one hyperedge, then either it has an internal node, or its only hyperedge is terminal, since otherwise $A\to D$ would be a chain production. Thus every application strictly increases $\mu$, except when $D$ has a single terminal hyperedge and no internal nodes. An application of this kind leaves $\mu$ unchanged and introduces a terminal hyperedge that remains in the final hypergraph. Hence there are at most $|E_{H_0}|$ such applications, and the total length of a derivation is at most
	\[
	B\eqdef\mu(H_0)+|E_{H_0}|.
	\]

	Therefore, to solve the uniform membership problem, we consider all derivations of length at most $B$. Since $\mu$ never decreases, we can discard a derivation as soon as its current hypergraph has measure greater than $\mu(H_0)$. At each step, we choose a production from $P^\prime$ and a nonterminal hyperedge to which it is applied; there are at most $|P^\prime|(\mu(H_0)+1)$ choices. Thus the number of derivations considered is at most
	\[
	\bigl(1+|P^\prime|(\mu(H_0)+1)\bigr)^B
	\le 2^{\mathit{poly}(|\Gamma|+|H_0|)}.
	\]
	For each terminal hypergraph obtained, check whether it is isomorphic to $H_0$ by enumerating bijections of nodes and hyperedges and checking that they preserve labels, attachment functions, and the external function. This also takes exponential time. Hence the uniform membership problem for HRGs is in \EXPTIME{}.
\end{proof}

The upper bound for repetition-free HRGs is more subtle. In presence of chain rules, it is not obvious that the size of a derivation of a hypergraph is linearly bounded by its size. To overcome this issue, we use the rational subset membership problem for symmetric groups.
\begin{proof}[Proof of \cref{theorem:upper_bound_NP}]
	Let $(\Gamma,H_0)$ be an input where $\Gamma=\langle N,T,\Sigma,P,S\rangle$ is a repetition-free HRG. We assume again that $\type(S)=\emptyset$ and that $H_0$ has at least one hyperedge, using the transformation described in the preceding proof. First, we do steps \ref{step-atmosttwo} and \ref{step-elimempty} of that proof, which require only polynomial time. Let $\hat\Gamma=\langle\hat N,T,\hat\Sigma,\hat P,\hat S\rangle$ be the resulting grammar and let $R$ be its order. $\hat \Gamma$ is equivalent to $\Gamma$, it is repetition-free, it does not have empty productions, and the right-hand side of each production has at most two hyperedges. Our goal is to deal with chain productions in $\hat P$. For each nonterminal label $A$, rename its selectors so that $\type(A)=\{1,\ldots,|\type(A)|\}$, making the same renaming in the attachment functions of $A$-labeled hyperedges and in the external functions of right-hand sides of productions with left-hand side $A$.

	Any chain production in $\hat\Gamma$ is either of the form $X\to F(Y,\UniSelA)$ for some bijection $\UniSelA$, or it is of the form $X\to D$ where the only nonterminal hyperedge has fewer attachment nodes than $D$ has external nodes (cf. Definition \ref{definition_chain}). Let us call productions of the first kind \emph{permutative}. Given $X,Y\in\hat N$ such that $\type(X)=\type(Y)=[t]$, let us define the NFA $\mathcal A(X,Y)=\langle\hat N,\Delta_{X,Y},\{X\},\{Y\}\rangle$, where $\Delta_{X,Y}$ consists of triples $(W,\UniSelA,Z)$ such that $W\to F(Z,\UniSelA)$ is a permutative production in $\hat P$ and $\type(W)=[t]$. The alphabet $\mathcal F\subseteq\Sym(t)$ over which this automaton is defined consists of functions $\UniSelA$ such that $(W,\UniSelA,Z)\in\Delta_{X,Y}$ for some $W,Z$. The size of $\mathcal A(X,Y)$ is polynomial w.r.t. the size of $\Gamma$.
	Since \textsc{RatSym} is in \NP{}, there is a nondeterministic polynomial algorithm $\mathcal T$ that takes $\mathcal A(X,Y)$ and $\UniSelA\in\Sym(t)$ as an input and answers whether there are $\UniSelA_1,\ldots,\UniSelA_k$ such that $\UniSelA=\UniSelA_1\cdot\dotsc\cdot\UniSelA_k$ and $\UniSelA_1\dotsc\UniSelA_k$ is accepted by $\mathcal A(X,Y)$. If $t=0$, the only permutation is the identity and this condition is checked by testing reachability in the automaton.

	Let $H_0\in\Language(\hat\Gamma)$. Let us fix a derivation of $H_0$ and consider the corresponding derivation tree. Suppose that, along some branch of the derivation tree, there is a sequence of chain productions $p_1,\ldots,p_k$ ($p_i=X_i\to D_i$), applied consecutively. Let us call such a sequence of productions a \emph{chain sequence} if it is maximal, i.e.~it is included in no larger sequence of chain productions along that branch.

	Some of the productions among $p_1,\ldots,p_k$ are non-permutative. Let us divide the chain sequence into blocks consisting of permutative productions separated by non-permutative ones:
	\[
	p_1,\ldots,p_k=\pi^1,r_1,\pi^2,\ldots,r_{m-1},\pi^m.
	\]
	Here $\pi^i=p_1^i,\ldots,p_{l_i}^i$ is a possibly empty sequence of permutative productions and $r_1,\ldots,r_{m-1}$ are non-permutative productions ($m\ge1$). Now, informally, the idea is to replace each nonempty $\pi^i$ in the derivation tree with a polynomial-size certificate given by an \NP{} algorithm for the rational subset membership problem.

	To be more precise, suppose that $l_i>0$, and denote the hyperedge to which the first rule in $\pi^i$ is applied by $e^\prime_{i-1}$. It holds that $p_j^i=X_j^i\to F(X_{j+1}^i,\UniSelA_j^i)$ for some $X_j^i,X_{j+1}^i,\UniSelA_j^i$, since all these rules are permutative. By \cref{lemma_composition}, applying the rules of $\pi^i$ results in replacing $e^\prime_{i-1}$ by $F(Y^i,\UniSelA^i)$, where $Y^i\eqdef X_{l_i+1}^i$ and $\UniSelA^i\eqdef\UniSelA_1^i\cdot\dotsc\cdot\UniSelA_{l_i}^i$. Then, in the fixed derivation tree, replace this sequence of production applications by a single application labeled by $(X_1^i,Y^i,\UniSelA^i)$. To verify this application, run $\mathcal T$ on the input $\mathcal A(X_1^i,Y^i)$ and $\UniSelA^i$.

	Note that the number $m-1$ of non-permutative productions in a chain sequence does not exceed $R$. Indeed, each such production strictly decreases the number of selectors of the only nonterminal hyperedge, since its attachment function is injective. Thus each chain sequence contributes at most $R+(R+1)=2R+1$ applications to the modified derivation tree.

	Recall the measure $\mu(H)=2|V_H|+|E_H|$ from the preceding proof. Applying any production from $\hat\Gamma$ does not decrease $\mu$. Each non-chain production strictly increases $\mu$, except when its right-hand side has a single terminal hyperedge and no internal nodes. There are at most $|E_{H_0}|$ applications of the latter kind, since the terminal hyperedges they introduce remain in the final hypergraph. Consequently, the number of non-chain applications is at most $b\eqdef\mu(H_0)+|E_{H_0}|$. Each chain sequence begins either at the initial hyperedge or at a nonterminal hyperedge introduced by a non-chain application. Since right-hand sides have at most two hyperedges, there are at most $1+2b$ chain sequences. Therefore, the total number of applications in the modified derivation tree is at most
	\[
	b+(1+2b)(2R+1),
	\]
	which is polynomial w.r.t. the sizes of $\Gamma$ and $H_0$. Therefore, the resulting derivation tree, together with the certificates for its permutative blocks and an isomorphism of its terminal hypergraph with $H_0$, is a polynomial-size certificate verifiable in polynomial time. Thus the problem is in \NP{}.
\end{proof}

\section{Lower Bounds}\label{section:lower-bounds}

We proceed with proving the lower bounds. It turns out that they are attained already for string-generating HRGs:

\begin{theorem}\label{theorem:lower_bound_EXPTIME}
	The uniform membership problem for string-generating HRGs is \EXPTIME{}-hard.
\end{theorem}
\begin{theorem}\label{theorem:lower_bound_NP}
	The uniform membership problem for repetition-free string-generating HRGs is \NP{}-hard.
\end{theorem}

One might expect that \cref{theorem:lower_bound_EXPTIME} could be proved by a reduction from the uniform membership problem for some other mildly context-sensitive grammar formalism. Table \ref{table:uniform-membership-mcsgf} suggests two candidates, namely, multiple context-free grammars (MCFGs) and deterministic tree-walking transducers (DTWTs) since the corresponding problem is \EXPTIME{}-complete for both of them. A transformation of DTWTs to string-generating HRGs is presented in \cite[Lemma 5.5]{EngelfrietH91}. It, however, requires exponential time since, in an HRG constructed from a DTWT according to this procedure, nonterminals range over tuples of even length $\le 2s$ that consist of DTWT's states ($s$ is the crossing number of the DTWT). Hence, if $Q$ is the set of states of the DTWT, the construction considers $\ge |Q|^{2j}$ tuples. In fact, there is no polynomial-time transformation of a DTWT (or an MCFG) into an equivalent string-generating HRG as the following remark shows.

\begin{remark}\label{remark:one-letter-sg-HRG}
	The uniform membership problem for string-generating HRGs over a one-letter alphabet $T=\{a\}$ is in P. Indeed, given a string-generating HRG $\Gamma$ and a word $w \in \{a\}^\ast$, checking whether $\SG(w) \in \Language(\Gamma)$ is equivalent to checking whether $w \in \Language(G)$ for any linearisation $G$ of $\Gamma$. Indeed, if $w = a^n$ is generated by $G$, then one can transform a derivation of $w$ in $G$ into the corresponding one in $\Gamma$ and obtain a derivation of some hypergraph $H$ which has $n$ $a$-labeled edges. However, since $\Gamma$ is string-generating, $H$ must equal $\SG(u)$ for some $u$; hence, clearly, $u=a^n$. The promise that $\Gamma$ is string-generating is crucial here as it guarantees uniqueness of a hypergraph with $n$ $a$-labeled edges in $\Language(\Gamma)$.
\end{remark}

Since checking whether the empty string $\varepsilon$ is accepted by a DTWT/MCFG is already \EXPTIME{}-complete \cite{BjorklundE13,KajiNSK94}, there is no polynomial-time transformation of a DTWT/MCFG into an equivalent string-generating HRG. However, in what follows, we shall use MCFGs, so let us recall their definition \cite{Kallmeyer10}.

\begin{definition}[Mcf-function]
	An \emph{mcf-function} is a function $f:(T^\ast)^{d_1} \times \ldots \times (T^\ast)^{d_k} \to (T^\ast)^{d_0}$ such that 
	\begin{equation}\label{equation:mcf-function}
		f((x_{11},\ldots,x_{1d_1}),\ldots,(x_{k1},\ldots,x_{kd_k}))
		=
		(y_{11}\ldots y_{1l_1},\ldots,y_{d_01}\ldots y_{d_0l_{d_0}})
	\end{equation}
	where each $y_{ij}$ (for $i=1,\ldots,d_0$ and $j=1,\ldots,l_{i}$) is either a terminal symbol from $T$ or it is one of the variables $x_{pq}$; besides, it is required that each variable $x_{pq}$ occurs \emph{at most once} among $y_{ij}$.
	
	Let us call a variable $x_{pq}$ \emph{unused in $f$} if it does not occur among $y_{ij}$. The \emph{information-lossless condition} is the property that each variable $x_{pq}$ occurs \emph{exactly once} among $y_{ij}$, i.e.~there are no variables unused in $f$.
\end{definition}

This and further notions are illustrated in \cref{example:MCFG-LCFRS}.

\begin{definition}
	An mcf-function is called \emph{$\varepsilon$-free} if no component in the right-hand side of \cref{equation:mcf-function} is the empty string, i.e.~$l_i>0$ for each $i=1,\ldots,d_0$. An mcf-function is called \emph{interleaving} if no component in the right-hand side of \cref{equation:mcf-function} contains a substring $x_{pq}x_{pr}$ for some $1 \le p \le k$ and some $1 \le q,r \le d_p$, $q \ne r$.
\end{definition}

\begin{definition}[Multiple context-free grammar]
	A \emph{multiple context-free grammar (MCFG)} is a tuple $G = \langle N,T,F,P,S\rangle$ where
	\begin{enumerate}
		\item $N$ is a finite set of nonterminals equipped with a function $\dim:N \to \Nat \setminus \{0\}$;
		\item $T$ is a finite set of terminals disjoint from $N$;
		\item $F$ is a finite set of mcf-functions;
		\item $P$ is a finite set of rules of the form $A_0 \to f[A_1,\ldots,A_k]$ where \\ $f:(T^\ast)^{d_1} \times \ldots \times (T^\ast)^{d_k} \to (T^\ast)^{d_0}$ is an mcf-function from $F$ and $A_i \in N$ with $\dim(A_i) = d_i$ for $i=0,\ldots,k$.
		\item $S \in N$ is the start symbol such that $\dim(S) = 1$. 
	\end{enumerate}
	
	Let $(\mathrm{yield}_G(A))_{A\in N}$ be the least tuple of sets such that, if $A \to f[A_1,\ldots,A_k]$ belongs to $P$ and $\tau_i \in \mathrm{yield}_G(A_i)$ for $i=1,\ldots,k$, then $f(\tau_1,\ldots,\tau_k) \in \mathrm{yield}_G(A)$. (In particular, for $k=0$, if $A \to f[]$ is a production with the right-hand side consisting of terminal symbols only, then $f() \in \mathrm{yield}_G(A)$.)
	The language generated by $G$ is $\mathrm{yield}_G(S)$.
\end{definition}

\begin{definition}
	A \emph{linear context-free rewriting system (LCFRS)} is an MCFG such that all mcf-functions used in it satisfy the information-lossless condition.
\end{definition}

\begin{definition}
	An LCFRS is \emph{repetition-free} if all mcf-functions participating in it are $\varepsilon$-free and interleaving.
\end{definition}

\begin{example}\label{example:MCFG-LCFRS}
	Let $G_1$ be an MCFG with nonterminals $S,A,B$ where $\dim(S) = 1$, $\dim(A) = 3$, $\dim(B) = 2$, and the productions are as follows:
	\begin{enumerate}
		\item $S \to f_1[A,B]$ where $f_1((x_{11},x_{12},x_{13}),(x_{21},x_{22})) = x_{11}x_{21}x_{12}x_{22}x_{13}$;
		\item $A \to f_2[A]$ where $f_2((x_{11},x_{12},x_{13})) = (ax_{11},bx_{12},cx_{13})$;
		\item $A \to f_3[]$ where $f_3() = (\varepsilon,\varepsilon,\varepsilon)$;
		\item $B \to f_4[A]$ where $f_4((x_{11},x_{12},x_{13})) = (dx_{11}d,dx_{12}d)$.
	\end{enumerate}
	This grammar generates the language $\{a^mda^ndb^mdb^ndc^m \mid m,n \in \Nat\}$. It is not an LCFRS because the last rule violates the information-lossless condition (there is no $x_{13}$ in the right-hand side). 
	
	Let $G_2$ be an MCFG in which the last rule is replaced with the rule
	$B \to f_5[A]$ where $f_5((x_{11},x_{12},x_{13})) = (dx_{11}d,dx_{12}x_{13}d)$. This is an LCFRS generating the language $\{a^mda^ndb^mdb^nc^ndc^m \mid m,n \in \Nat\}$.
	
	The functions $f_1,f_2,f_4,f_5$ are $\varepsilon$-free; the functions $f_1,f_2,f_3,f_4$ are interleaving. $f_5$ is not interleaving because it contains the substring $x_{12}x_{13}$.
\end{example}

It is straightforward to prove the following proposition.
\begin{proposition}\label{proposition:lcfrs-to-hrg}
	There is a polynomial-time transformation of an LCFRS $G$ into an HRG $\Gamma$ such that $\Language(\Gamma) = \SG(\Language(G))$. Besides, if $G$ is repetition-free, then $\Gamma$ is repetition-free.
\end{proposition}

Let us first show an example that explains how to construct $\Gamma$ from $G$.
\begin{example}
	The LCFRS $G_2$ from \cref{example:MCFG-LCFRS} is transformed into the HRG $\Gamma$ with nonterminals $S,A,B$ ($\type(S)=[2]$, $\type(A) = [6]$, $\type(B) = [4]$) that has the following productions:
	\begin{enumerate}
		\item 
		$
		S \;\to\;
		\vcenter
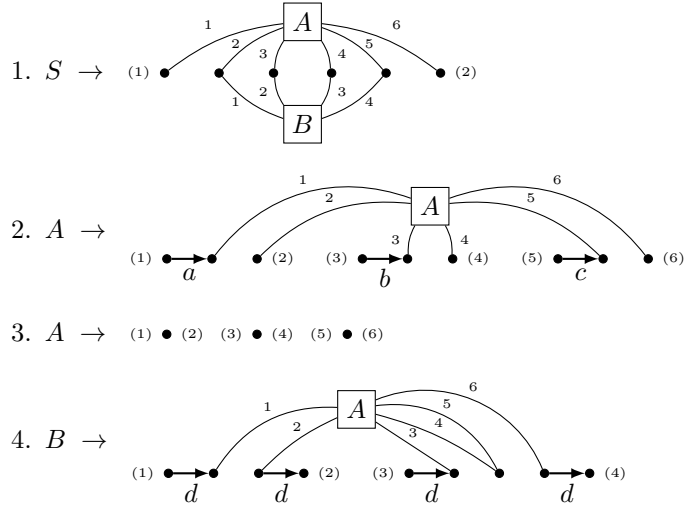
{\hbox{{\tikz[baseline=.1ex]{
						\node (C) {};
						\node[node,left=2mm of C] (V3) {};
						\node[node,left=6mm of V3] (V2) {};
						\node[node,left=6mm of V2,label=left:{\tiny $(1)$}] (V1) {};
						\node[node,right=2mm of C] (V4) {};
						\node[node,right=6mm of V4] (V5) {};
						\node[node,right=6mm of V5,label=right:{\tiny $(2)$}] (V6) {};
						\node[hyperedge,above=3mm of C] (E1) {$A$};
						\node[hyperedge,below=3mm of C] (E2) {$B$};
						\draw[-] (V1) to[bend left=15] node[above left] {\tiny $1$} (E1);
						\draw[-] (V2) to[bend left=15] node[left] {\tiny $2$} (E1);
						\draw[-] (V3) to[bend left=15] node[left] {\tiny $3$} (E1);
						\draw[-] (V4) to[bend right=15] node[right] {\tiny $4$} (E1);
						\draw[-] (V5) to[bend right=15] node[right] {\tiny $5$} (E1);
						\draw[-] (V6) to[bend right=15] node[above right] {\tiny $6$} (E1);
						\draw[-] (V2) to[bend right=15] node[left] {\tiny $1$} (E2);
						\draw[-] (V3) to[bend right=15] node[left] {\tiny $2$} (E2);
						\draw[-] (V4) to[bend left=15] node[right] {\tiny $3$} (E2);
						\draw[-] (V5) to[bend left=15] node[right] {\tiny $4$} (E2);
		}}}}
		$
		\item
		$
		A \;\to\;
		\vcenter{\hbox{{\tikz[baseline=.1ex]{
						\def\STEP{6mm}
						\def\SHIFT{8mm}
						\node[hyperedge] (EA) at ($(\SHIFT+\STEP*5.5,0.7)$) {$A$};
						\foreach \i in {1,...,9}
						{
							\pgfmathtruncatemacro\MOD{mod(\i,3)}
							\pgfmathtruncatemacro\DIV{div(\i-1,3)}
							\pgfmathtruncatemacro\NUMLEFT{1+(2*div(\i,3))}
							\pgfmathtruncatemacro\NUMRIGHT{2*div(\i,3)}
							\pgfmathtruncatemacro\BEND{7*pow((\i-5.5)/2,2)+15}
							\node[node,
							label=
							\ifnumgreater{\MOD}{0}{left}{right}:
							{\tiny
								\ifnumequal{\MOD}{1}{$(\NUMLEFT)$}{\ifnumequal{\MOD}{0}{$(\NUMRIGHT)$}{}}
							}
							] (V\i) at ($(\STEP*\i+\DIV*\SHIFT,0)$) {};
							\ifnumequal{\MOD}{2}
							{
								\draw[-] (EA) to[bend \ifnumgreater{\i}{5}{left}{right}=\BEND] node[\ifnumequal{\NUMLEFT}{3}{left}{above}] {\tiny $\NUMLEFT$} (V\i);
							}{}
							\ifnumequal{\MOD}{0}
							{
								\draw[-] (EA) to[bend \ifnumgreater{\i}{5}{left}{right}=\BEND] node[\ifnumequal{\NUMRIGHT}{4}{right}{above}] {\tiny $\NUMRIGHT$} (V\i);
							}{}
						}
						\draw[thick,-latex] (V1) -- node[below] {$a$} (V2);
						\draw[thick,-latex] (V4) -- node[below] {$b$} (V5);
						\draw[thick,-latex] (V7) -- node[below] {$c$} (V8);
		}}}}
		$
		\item\label{item:3}
		$
		A \;\to\;
		\vcenter{\hbox{{\tikz[baseline=.1ex]{
						\def\STEP{12mm}
						\foreach \i in {1,...,3}
						{
							\pgfmathtruncatemacro\LEFT{2*\i-1}
							\pgfmathtruncatemacro\RIGHT{2*\i}
							\node[node,
							label=left:{\tiny $(\LEFT)$},
							label=right:{\tiny $(\RIGHT)$}
							]
							(V\i) at ($(\STEP*\i,0)$) {};
						}
		}}}}
		$
		\item\label{item:4}
		$
		B \;\to\;
		\vcenter{\hbox{{\tikz[baseline=.1ex]{
						\def\STEP{6mm}
						\def\SHIFT{8mm}
						\node[hyperedge] (EA) at ($(\SHIFT*0.5+\STEP*4.5,0.85)$) {$A$};
						\foreach \i in {1,...,4}
						{
							\node[node,
							label=
							\ifnumgreater{\i}{1}{right}{left}:
							{\tiny
								\ifnumequal{\i}{1}{$(1)$}{\ifnumequal{\i}{4}{$(2)$}{}}
							}
							] (V\i) at ($(\STEP*\i,0)$) {};
						}
						\foreach \i in {5,...,9}
						{
							\node[node,
							label=
							\ifnumgreater{\i}{5}{right}{left}:
							{\tiny
								\ifnumequal{\i}{5}{$(3)$}{\ifnumequal{\i}{9}{$(4)$}{}}
							}
							] (V\i) at ($(\STEP*\i+\SHIFT,0)$) {};
						}
						\draw[thick,-latex] (V1) -- node[below] {$d$} (V2);
						\draw[thick,-latex] (V3) -- node[below] {$d$} (V4);
						\draw[thick,-latex] (V5) -- node[below] {$d$} (V6);
						\draw[thick,-latex] (V8) -- node[below] {$d$} (V9);
						\draw[-] (EA) to[bend right=30] node[above] {\tiny $1$} (V2);
						\draw[-] (EA) to[bend right=10] node[above] {\tiny $2$} (V3);
						\draw[-] (EA) to[bend left=0] node[above] {\tiny $3$} (V6);
						\draw[-] (EA) to[bend left=10] node[above] {\tiny $4$} (V7);
						\draw[-] (EA) to[bend left=35] node[above] {\tiny $5$} (V7);
						\draw[-] (EA) to[bend left=45] node[above] {\tiny $6$} (V8);
		}}}}
		$
	\end{enumerate}
	Since $f_3$ has three $\varepsilon$'s in the right-hand side, the hypergraph from \cref{item:3} has three pairs of selectors such that, in each pair, selectors are mapped to the same node. The resulting hypergraph is not repetition-free. Since $f_5$ has the substring $x_{12}x_{13}$, some attachment nodes of the $A$-labeled hyperedge of the hypergraph from \cref{item:4} coincide.
\end{example}

\begin{proof}
	Let $G=\langle N,T,F,P,S\rangle$ be an LCFRS. We construct
	a HRG $\Gamma$ with the same nonterminal and terminal alphabets
	and the same start symbol. Set
	\[
	\type(A)=[2\dim(A)] \quad (A\in N),
	\qquad
	\type(a)=[2] \quad (a\in T).
	\]
	Choose a finite selector set $\Sigma$ containing all these types.
	
	First, extend the string-graph notation to tuples. For
	$\tau=(w_1,\ldots,w_d)$, let $\SG^{(d)}(\tau)$ be the disjoint
	union of copies of $\SG(w_1),\ldots,\SG(w_d)$, with external
	selectors $2i-1$ and $2i$ assigned to the initial and final
	nodes of the $i$-th copy, respectively. In particular, if
	$w_i=\varepsilon$, this copy consists of one node to which
	both selectors $2i-1$ and $2i$ are mapped. Copies belonging to different
	components have disjoint node sets, including when their
	words are empty. Thus $\SG^{(1)}((w))=\SG(w)$.
	
	Consider a rule
	\(
	\rho = A_0\to f[A_1,\ldots,A_k]
	\)
	where $f$ is an mcf-function as in \eqref{equation:mcf-function}. Define the right-hand side $H_\rho$ of the corresponding HRG rule as follows. Its nodes are
	\[
	V_{H_\rho}
	=
	\{v_{i,j}\mid 1\le i\le d_0,\ 0\le j\le l_i\},
	\]
	all pairwise distinct. Thus the $i$-th output component has
	one node at every boundary between consecutive symbols,
	including its initial and final boundaries.
	
	For every terminal $y_{ij}=a\in T$, introduce
	a new edge $t_{i,j}$ with
	\(
	\lab_{H_\rho}(t_{i,j})=a\),	\(\att_{H_\rho}(t_{i,j})(1)=v_{i,j-1}\), \(\att_{H_\rho}(t_{i,j})(2)=v_{i,j}\).
	For each $p=1,\ldots, k$, introduce one
	nonterminal hyperedge $e_p$ labelled $A_p$. Whenever
	$y_{ij}=x_{pq}$, define
	\[
	\att_{H_\rho}(e_p)(2q-1)=v_{i,j-1},
	\qquad
	\att_{H_\rho}(e_p)(2q)=v_{i,j}.
	\]
	These assignments define the attachment function of $e_p$ on
	its entire type $[2d_p]$, because every variable $x_{pq}$
	occurs exactly once. There are no other
	hyperedges.
	
	Set $\type(H_\rho)=[2d_0]$ and define
	\(
	\ext_{H_\rho}(2i-1)=v_{i,0}\),
	\(\ext_{H_\rho}(2i)=v_{i,l_i}
	\quad (1\le i\le d_0)\). If $l_i=0$, these two external selectors designate the same node.
	
	Finally, let 
	\(
	P_\Gamma=\{A_0\to H_\rho\mid
	\rho=(A_0\to f[A_1,\ldots,A_k])\in P\}\) and \(\Gamma=\langle N,T,\Sigma,P_\Gamma,S\rangle\).
	
	The construction satisfies the following substitution
	identity. For arbitrary tuples
	$\tau_p=(w_{p1},\ldots,w_{pd_p})$, we have
	\[
	H_\rho[e_1/\SG^{(d_1)}(\tau_1)]\ldots [e_k/\SG^{(d_k)}(\tau_k)]
	=
	\SG^{(d_0)}\bigl(f(\tau_1,\ldots,\tau_k)\bigr)
	\]
	Indeed, the component representing $w_{pq}$ is inserted
	between the two boundary nodes assigned to the unique
	occurrence of $x_{pq}$. Terminal occurrences already have
	their corresponding terminal edges. Hence, for each $i$,
	these paths and terminal edges concatenate in the order
	specified by $y_{i1}\cdots y_{il_i}$. If $w_{pq}=\varepsilon$, its string graph has a single node,
	so replacement identifies the two boundary nodes of the
	corresponding variable occurrence. This is exactly the
	effect of substituting an empty word.
	
	Consequently, a straightforward induction shows that, for every $A\in N$,
	\[
	\{H\in\mathcal H(T)\mid A^\bullet
	\Rightarrow_\Gamma^* H\}
	=
	\{\SG^{(\dim(A))}(\tau)\mid
	\tau\in\mathrm{yield}_G(A)\},
	\]
	Since $\dim(S)=1$, we conclude that
	\(
	\Language(\Gamma)=\SG(\Language(G))
	\), as desired.
	
	Suppose now that every function used in $G$ is $\varepsilon$-free
	and interleaving. Fix a production $A_0\to H_\rho$. First, $\varepsilon$-freeness gives $l_i>0$ for every output
	component. Hence $v_{i,0}\ne v_{i,l_i}$ and therefore $\ext_{H_\rho}$ is injective. Secondly, every terminal-labelled edge has distinct attachment nodes
	$v_{i,j-1}$ and $v_{i,j}$. It remains to check a nonterminal
	hyperedge $e_p$. Its attachment nodes can coincide iff $\att_{H_\rho}(e_p)(2q) = \att_{H_\rho}(e_p)(2r-1)$ for some $q \ne r$. This implies that the definition of $f$ contains the substring \(x_{pq}x_{pr}\). This contradicts the
	interleaving condition. Consequently,
	$\att_{H_\rho}(e_p)$ is injective.
\end{proof}

Since the uniform membership for LCFRSs is \PSPACE{}-complete, \cref{proposition:lcfrs-to-hrg} implies that the uniform membership for string-generating HRGs is \PSPACE{}-hard\footnote{Note that we proved \PSPACE{}-hardness for general HRGs in \cref{section:chain} using different ideas, namely, via transformation monoids.}. However, we cannot use a similar transformation to convert an arbitrary MCFG into an HRG because HRGs are information-lossless: it is not possible to apply a rule $A \to H$ to a hyperedge $e$ only partially. 

Despite this fact, there is an indirect way of using the \EXPTIME{}-hardness result for MCFGs to prove \EXPTIME{}-hardness of the uniform membership for string-generating HRGs. First, we need the following proposition.

\begin{proposition}\label{proposition:EXPTIME-LCFRS}
	The following problem is \EXPTIME{}-hard:
	\begin{quote}
		\textsc{Instance.} An LCFRS $G$ in the alphabet $\{a,b\}$ such that $\varepsilon \notin \Language(G)$.
		\\
		\textsc{Question.} Is the intersection $\Language(G) \cap b\{a,b\}^\ast$ non-empty?
	\end{quote}
\end{proposition}
\begin{proof}
	In \cite[Appendix A.1]{KajiNSK94}, it is proved that the uniform membership for MCFGs is \EXPTIME{}-hard. In fact, the authors prove a stronger result, namely, \EXPTIME{}-hardness of the following problem:
	\begin{quote}
		\textsc{Instance.} An MCFG $G$ in the alphabet $\{a\}$.
		\\
		\textsc{Question.} Does $\varepsilon$ belong to $\Language(G)$?
	\end{quote}
	We reduce the latter problem to the one defined in the proposition. Let $G = \langle N,\{a\},F,P,S\rangle$ be an MCFG. We construct an LCFRS $\tilde{G} = \langle \tilde{N},\{a,b\},\tilde{F},\tilde{P},\tilde{S}\rangle$:
	\begin{itemize}
		\item $\tilde{N} \eqdef N \cup \{\tilde{S}\}$ with $\tilde S$ fresh, and with the new dimension function $\widetilde{\dim}(A) \eqdef \dim(A)+1$ for $A \in N$ and $\widetilde{\dim}(\tilde{S}) \eqdef 1$.
		\item For $T = \{a,b\}$, let $f:(T^\ast)^{d_1} \times \ldots \times (T^\ast)^{d_k} \to (T^\ast)^{d_0}$ be an mcf-function from $F$. We define an mcf-function $\tilde{f}:(T^\ast)^{d_1+1} \times \ldots \times (T^\ast)^{d_k+1} \to (T^\ast)^{d_0+1}$ as follows. For $i=1,\ldots,d_0$, the $i$-th component of 
		$$
		\tilde{f}\left(\left(x_{11},\ldots,x_{1d_1},x_{1(d_1+1)}\right),\ldots,\left(x_{k1},\ldots,x_{kd_k},x_{k(d_k+1)}\right)\right)
		$$
		is the same as that of $f(\left(x_{11},\ldots,x_{1d_1}\right),\ldots,\left(x_{k1},\ldots,x_{kd_k}\right))$, and its $(d_0+1)$-st component equals
		$
		x_{1(d_1+1)}x_{2(d_2+1)}\ldots x_{k(d_k+1)} z_1 \ldots z_l
		$
		where $z_1,\ldots,z_l$ are all the variables unused in $f$. (If $k=0$, then this component equals $\varepsilon$.) Clearly, the function $\tilde{f}$ satisfies the information-lossless condition.
		\\
		Besides, let $g:(T^\ast)^2 \to T^\ast$ be an mcf-function defined as follows: \[g(x_1,x_2) = x_1bx_2.\]
		
		\item $\tilde{P} \eqdef \{A \to \tilde{f}[A_1,\ldots,A_k] \mid (A \to f[A_1,\ldots,A_k]) \in P\} \cup \{\tilde{S} \to g[S]\}$.
	\end{itemize}
	It is not hard to see that, for each $A \in N$, $\mathrm{yield}_{\tilde{G}}(A)$ without the last component equals $\mathrm{yield}_G(A)$. In particular, the first components of tuples from $\mathrm{yield}_{\tilde{G}}(S)$ form exactly $\Language(G)$. Therefore, 
	\begin{align*}
		\varepsilon \in \Language(G)
		& \Longleftrightarrow \\
		(\varepsilon,w) \in \mathrm{yield}_{\tilde{G}}(S)
		~\text{for some $w \in \{a,b\}^\ast$}
		& \Longleftrightarrow \\
		g(\varepsilon,w) = bw \in \Language(\tilde{G})
		~\text{for some $w \in \{a,b\}^\ast$.}
	\end{align*}
	Finally, note that $\varepsilon \notin \Language(\tilde{G})$ because the rule $\tilde{S} \to g[S]$, which must be applied in any derivation starting with $\tilde S$, introduces the symbol $b$.
\end{proof}

The second ingredient in the proof of \EXPTIME{}-hardness we shall present soon is using generalised string graphs.
\begin{definition}[Generalised string graph]
	\label{definition:generalized-string-graph}
	For a non-empty finite selector set $\UniSetSelsB$, a $\UniSetSelsB$-string hypergraph $\SG_{\UniSetSelsB}(w)$ induced by a string $w = a_1 \ldots a_n$ is defined as follows: 
	\begin{itemize}
		\item $V_{\SG_{\UniSetSelsB}(w)} = \{v_{i,\UniSelB} \mid i=0,\ldots,n,\, \UniSelB \in \UniSetSelsB \}$; 
		\item $E_{\SG_{\UniSetSelsB}(w)} = \{s_1,\dotsc,s_n\}$; $\type(s_i) = \type(\SG_{\UniSetSelsB}(w)) = \Theta \times \{1,2\}$;
		\item $\att_{\SG_{\UniSetSelsB}(w)}(s_i)(\UniSelB,1)=v_{i-1,\UniSelB}$ and $\att_{\SG_{\UniSetSelsB}(w)}(s_i)(\UniSelB,2)=v_{i,\UniSelB}$ for $i=1,\ldots,n$ and $\UniSelB \in \UniSetSelsB$;
		\item $\lab_{\SG_{\UniSetSelsB}(w)}(s_i)=a_i$ for $i = 1, \ldots, n$; 
		\item $\ext_{\SG_{\UniSetSelsB}(w)}(\UniSelB,1)=v_{0,\UniSelB}$, $\ext_{\SG_{\UniSetSelsB}(w)}(\UniSelB,2)=v_{n,\UniSelB}$. 
	\end{itemize}
\end{definition}
If $\Theta = [q]$ for some $q \in \Nat \setminus \{0\}$, we identify the selector $(i,1)$ with $i$ and the selector $(i,2)$ with $i+q$.
\begin{example}
	Below, an example of a $[6]$-string hypergraph is presented.
	$$
	\SG_{[6]}(aba) \; = \; 
	\vcenter
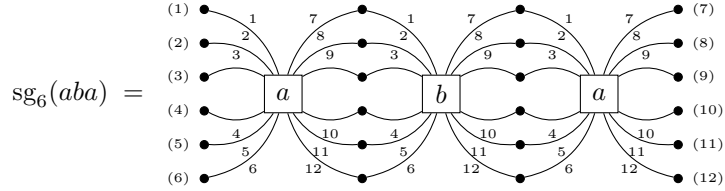
{\hbox{{\tikz[baseline=.1ex]{
					\def\HOR{0.45}
					\def\VER{2.1}
					\def\BEND{30}
					\node[hyperedge] (E1) at ($(\VER*0.5,-\HOR*2.5)$) {$a$};
					\node[hyperedge] (E2) at ($(\VER*1.5,-\HOR*2.5)$) {$b$};
					\node[hyperedge] (E3) at ($(\VER*2.5,-\HOR*2.5)$) {$a$};
					\foreach \i in {0,...,3}
					{
						\foreach \j in {0,...,5}
						{
							\pgfmathtruncatemacro\J{\j+1}
							\pgfmathtruncatemacro\JJ{\j+7}
							\node[node, label=\ifnumgreater{\i}{0}{right}{left}:{\tiny \ifnumequal{\i}{0}{$(\J)$}{\ifnumequal{\i}{3}{$(\JJ)$}{}}}] (V\i\j) at ($(\VER*\i,-\HOR*\j)$) {};
						}			
					}
					\foreach \i in {0,...,3}
					{
						\foreach \j in {0,...,5}
						{
							\pgfmathtruncatemacro\I{\i+1}
							\ifnumgreater{3}{\i}{
								\pgfmathtruncatemacro\J{\j+1}
								\draw[-] (E\I) to[bend \ifnumgreater{\j}{2}{left}{right}=\BEND] node[\ifnumgreater{\j}{2}{below}{above}] {\tiny $\J$} (V\i\j);
							}
							{}
							\ifnumgreater{\i}{0}{
								\pgfmathtruncatemacro\JJ{\j+7}
								\draw[-] (E\i) to[bend \ifnumgreater{\j}{2}{right}{left}=\BEND] node[\ifnumgreater{\j}{2}{below}{above}] {\tiny $\JJ$} (V\i\j);
							}
							{}
						}
					}
	}}}}
	$$
\end{example}
Clearly, $[1]$-string hypergraphs are exactly string graphs. It is also clear that one can generalise \cref{proposition:lcfrs-to-hrg} to $\Theta$-string hypergraphs as follows.
\begin{proposition}\label{proposition:lcfrs-to-hrg-generalised}
	For each fixed non-empty finite set $\UniSetSelsB$ of selectors, there is a polynomial-time transformation of an LCFRS $G$ into an HRG $\Gamma$ such that $\Language(\Gamma) = \SG_\UniSetSelsB(\Language(G))$. 
\end{proposition}
It is proved in the same way as \cref{proposition:lcfrs-to-hrg} with the only difference that each hyperedge has $|\UniSetSelsB|$ times more attachment nodes.

We are ready to prove the \EXPTIME{} lower bound.

\begin{proof}[Proof of \cref{theorem:lower_bound_EXPTIME}]
	We are going to reduce the \EXPTIME{}-hard problem from \cref{proposition:EXPTIME-LCFRS} to the uniform membership problem for string-generating HRGs. Let $G$ be an LCFRS over the terminal alphabet $\{a,b\}$ not generating the empty word; our task is to check whether there is a word starting with $b$ in $G$.
	
	Construct an HRG $\Gamma = \langle N,\{a,b\},\Sigma,P,S\rangle$ such that $\Language(\Gamma) = \SG_{[6]}(\Language(G))$ (\cref{proposition:lcfrs-to-hrg-generalised}) and define $\Gamma^\prime = \langle N^\prime, \{x,y\}, \Sigma, P^\prime, S^\prime \rangle$ where $N^\prime = N \cup \{a,b,S^\prime\}$ with $S^\prime$ fresh, $\type(S^\prime)=[2]$, and $\type(x)=\type(y)=[2]$, and $P^\prime = P \cup \{a \to H_a, b \to H_b, S^\prime \to H_0\}$. Hypergraphs $H_a,H_b,H_0$ are defined below.
	\begin{itemize}
		\item $V_{H_a} = \{v_1,v_2,v_3\}$, $E_{H_a} = \emptyset$, $\type(H_a)=[12]$, $\ext_{H_a}(i) = v_1$ for $i=1,2$, $\ext_{H_a}(i) = v_2$ for $i=3,4$, $\ext_{H_a}(i) = v_3$ for $i=5,\ldots,12$;
		\item $V_{H_b} = \{v_1,v_2,v_3\}$, $E_{H_b} = \emptyset$, $\type(H_b)=[12]$, $\ext_{H_b}(i) = v_1$ for $i=1,4$, $\ext_{H_b}(i) = v_2$ for $i=2,5$, $\ext_{H_b}(i) = v_3$ for $i=3,6$ as well as for $i=7,\ldots,12$.
		\item 
		$
		H_0 \;=\;
		\vcenter
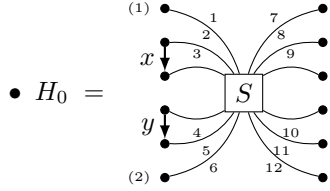
{\hbox{{\tikz[baseline=.1ex]{
						\def\HOR{0.45}
						\def\VER{2.1}
						\def\BEND{30}
						\node[hyperedge] (E1) at ($(\VER*0.5,-\HOR*2.5)$) {$S$};
						\foreach \i in {0,1}
						{
							\foreach \j in {0,...,5}
							{
								\pgfmathtruncatemacro\J{\j+1}
								\pgfmathtruncatemacro\JJ{\j+7}
								\node[node, label=left:{\tiny
									\ifnumequal{\i}{0}{\ifnumequal{\j}{0}{$(1)$}{\ifnumequal{\j}{5}{$(2)$}{}}}{}
								}] (V\i\j) at ($(\VER*\i,-\HOR*\j)$) {};
							}			
						}
						\foreach \i in {0,1}
						{
							\foreach \j in {0,...,5}
							{
								\pgfmathtruncatemacro\I{\i+1}
								\ifnumgreater{1}{\i}{
									\pgfmathtruncatemacro\J{\j+1}
									\draw[-] (E\I) to[bend \ifnumgreater{\j}{2}{left}{right}=\BEND] node[\ifnumgreater{\j}{2}{below}{above}] {\tiny $\J$} (V\i\j);
								}
								{}
								\ifnumgreater{\i}{0}{
									\pgfmathtruncatemacro\JJ{\j+7}
									\draw[-] (E\i) to[bend \ifnumgreater{\j}{2}{right}{left}=\BEND] node[\ifnumgreater{\j}{2}{below}{above}] {\tiny $\JJ$} (V\i\j);
								}
								{}
							}
						}
						\draw[thick,-latex] (V01) to[bend right=0] node[left] {$x$} (V02);
						\draw[thick,-latex] (V03) to[bend right=0] node[left] {$y$} (V04);
		}}}}
		$
	\end{itemize}
	Productions $a \to H_a$ and $b \to H_b$ are empty. For example, when $a \to H_a$ is applied to a hyperedge $e$ in a hypergraph $H$, it glues $\att_H(e)(1)$ with $\att_H(e)(2)$, $\att_H(e)(3)$ with $\att_H(e)(4)$ and, for $i=5,\ldots,12$, it glues the nodes $\att_H(e)(i)$. Similarly, if $\lab_H(e)=b$, then $b \to H_b$, being applied to $e$, glues $\att_H(e)(1)$ with $\att_H(e)(4)$, $\att_H(e)(2)$ with $\att_H(e)(5)$ and $\att_H(e)(3)$ with $\att_H(e)(i)$ for $i=6,\ldots,12$.
	
	\begin{example}\label{example:aba}
		Assume that $aba \in \Language(G)$; then, $S^\bullet \Rightarrow^\ast \SG_{[6]}(aba)$ in $\Gamma$. The following is then a derivation in $\Gamma^\prime$:
		$$
		(S^\prime)^\bullet
		\Rightarrow
		H_0
		\Rightarrow^{\ast}
		\vcenter{\hbox{{\tikz[baseline=.1ex]{
						\def\HOR{0.45}
						\def\VER{2.1}
						\def\BEND{30}
						\node[hyperedge] (E1) at ($(\VER*0.5,-\HOR*2.5)$) {$a$};
						\node[hyperedge] (E2) at ($(\VER*1.5,-\HOR*2.5)$) {$b$};
						\node[hyperedge] (E3) at ($(\VER*2.5,-\HOR*2.5)$) {$a$};
						\foreach \i in {0,...,3}
						{
							\foreach \j in {0,...,5}
							{
								\pgfmathtruncatemacro\J{\j+1}				\pgfmathtruncatemacro\JJ{\j+7}
								\node[node, label=left:{\tiny
									\ifnumequal{\i}{0}{\ifnumequal{\j}{0}{$(1)$}{\ifnumequal{\j}{5}{$(2)$}{}}}{}
								}] (V\i\j) at ($(\VER*\i,-\HOR*\j)$) {};
							}			
						}
						\foreach \i in {0,...,3}
						{
							\foreach \j in {0,...,5}
							{
								\pgfmathtruncatemacro\I{\i+1}
								\ifnumgreater{3}{\i}{
									\pgfmathtruncatemacro\J{\j+1}
									\draw[-] (E\I) to[bend \ifnumgreater{\j}{2}{left}{right}=\BEND] node[\ifnumgreater{\j}{2}{below}{above}] {\tiny $\J$} (V\i\j);
								}
								{}
								\ifnumgreater{\i}{0}{
									\pgfmathtruncatemacro\JJ{\j+7}
									\draw[-] (E\i) to[bend \ifnumgreater{\j}{2}{right}{left}=\BEND] node[\ifnumgreater{\j}{2}{below}{above}] {\tiny $\JJ$} (V\i\j);
								}
								{}
							}
						}
						\draw[thick,-latex] (V01) to[bend right=0] node[left] {$x$} (V02);
						\draw[thick,-latex] (V03) to[bend right=0] node[left] {$y$} (V04);
		}}}}
		\Rightarrow
		$$
		$$
		\vcenter{\hbox{{\tikz[baseline=.1ex]{
						\def\HOR{0.45}
						\def\VER{2.1}
						\def\SHIFT{0.6}
						\def\BEND{30}
						\node[node, label=left:{\tiny $(1)$}] (V1) at ($(\VER*0,-\HOR*0)$) {};
						\node[node] (V2) at ($(\VER*0,-\HOR*1.25)$) {};
						\node[node, label=left:{\tiny $(2)$}] (V3) at ($(\VER*0,-\HOR*2.5)$) {};
						\node[hyperedge] (E1) at ($(\SHIFT+\VER*0.5,-\HOR*2.5)$) {$b$};
						\node[hyperedge] (E2) at ($(\SHIFT+\VER*1.5,-\HOR*2.5)$) {$a$};
						\foreach \i in {1,2}
						{
							\foreach \j in {0,...,5}
							{
								\pgfmathtruncatemacro\J{\j+1}
								\pgfmathtruncatemacro\JJ{\j+7}
								\node[node, label=left:{\tiny
									\ifnumequal{\i}{0}{\ifnumequal{\j}{0}{$(1)$}{\ifnumequal{\j}{5}{$(2)$}{}}}{}
								}] (V\i\j) at ($(\SHIFT+\VER*\i,-\HOR*\j)$) {};
							}			
						}
						\foreach \i in {1,2}
						{
							\foreach \j in {0,...,5}
							{
								\pgfmathtruncatemacro\I{\i+1}
								\ifnumgreater{2}{\i}{
									\pgfmathtruncatemacro\J{\j+1}
									\draw[-] (E\I) to[bend \ifnumgreater{\j}{2}{left}{right}=\BEND] node[\ifnumgreater{\j}{2}{below}{above}] {\tiny $\J$} (V\i\j);
								}
								{}
								\ifnumgreater{\i}{0}{
									\pgfmathtruncatemacro\JJ{\j+7}
									\draw[-] (E\i) to[bend \ifnumgreater{\j}{2}{right}{left}=\BEND] node[\ifnumgreater{\j}{2}{below}{above}] {\tiny $\JJ$} (V\i\j);
								}
								{}
							}
						}
						\draw[thick,-latex] (V1) to[bend right=0] node[left] {$x$} (V2);
						\draw[thick,-latex] (V2) to[bend right=0] node[left] {$y$} (V3);
						\draw[-] (E1) to[bend right=72] node[above] {\tiny $1$} (V3);
						\draw[-] (E1) to[bend right=32] node[above] {\tiny $2$} (V3);
						\draw[-] (E1) to[bend right=5] node[above] {\tiny $3$} (V3);
						\draw[-] (E1) to[bend left=5] node[below] {\tiny $4$} (V3);
						\draw[-] (E1) to[bend left=32] node[below] {\tiny $5$} (V3);
						\draw[-] (E1) to[bend left=72] node[below] {\tiny $6$} (V3);
		}}}}
		\Rightarrow
		\vcenter{\hbox{{\tikz[baseline=.1ex]{
						\def\HOR{0.45}
						\def\VER{2.1}
						\def\SHIFT{0.6}
						\def\BEND{30}
						\node[node, label=left:{\tiny $(1)$}] (V1) at ($(\VER*0,-\HOR*0)$) {};
						\node[node] (V2) at ($(\VER*0,-\HOR*1.25)$) {};
						\node[node, label=left:{\tiny $(2)$}] (V3) at ($(\VER*0,-\HOR*2.5)$) {};
						\node[hyperedge] (E1) at ($(\SHIFT+\VER*0.5,-\HOR*2.5)$) {$a$};
						\foreach \i in {1}
						{
							\foreach \j in {0,...,5}
							{
								\pgfmathtruncatemacro\J{\j+1}
								\pgfmathtruncatemacro\JJ{\j+7}
								\node[node, label=left:{\tiny
									\ifnumequal{\i}{0}{\ifnumequal{\j}{0}{$(1)$}{\ifnumequal{\j}{5}{$(2)$}{}}}{}
								}] (V\i\j) at ($(\SHIFT+\VER*\i,-\HOR*\j)$) {};
							}			
						}
						\foreach \i in {1}
						{
							\foreach \j in {0,...,5}
							{
								\pgfmathtruncatemacro\I{\i+1}
								\ifnumgreater{1}{\i}{
									\pgfmathtruncatemacro\J{\j+1}
									\draw[-] (E\I) to[bend \ifnumgreater{\j}{2}{left}{right}=\BEND] node[\ifnumgreater{\j}{2}{below}{above}] {\tiny $\J$} (V\i\j);
								}
								{}
								\ifnumgreater{\i}{0}{
									\pgfmathtruncatemacro\JJ{\j+7}
									\draw[-] (E\i) to[bend \ifnumgreater{\j}{2}{right}{left}=\BEND] node[\ifnumgreater{\j}{2}{below}{above}] {\tiny $\JJ$} (V\i\j);
								}
								{}
							}
						}
						\draw[thick,-latex] (V1) to[bend right=0] node[left] {$x$} (V2);
						\draw[thick,-latex] (V2) to[bend right=0] node[left] {$y$} (V3);
						\draw[-] (E1) to[bend right=72] node[above] {\tiny $1$} (V3);
						\draw[-] (E1) to[bend right=32] node[above] {\tiny $2$} (V3);
						\draw[-] (E1) to[bend right=5] node[above] {\tiny $3$} (V3);
						\draw[-] (E1) to[bend left=5] node[below] {\tiny $4$} (V3);
						\draw[-] (E1) to[bend left=32] node[below] {\tiny $5$} (V3);
						\draw[-] (E1) to[bend left=72] node[below] {\tiny $6$} (V3);
		}}}}
		\Rightarrow
		\vcenter{\hbox{{\tikz[baseline=.1ex]{
						\def\HOR{0.45}
						\def\VER{2.1}
						\def\SHIFT{1}
						\def\BEND{30}
						\node[node, label=above:{\tiny $(1)$}] (V1) at ($(\VER*0,-\HOR*0)$) {};
						\node[node] (V2) at ($(\VER*0,-\HOR*1.25)$) {};
						\node[node, label=below:{\tiny $(2)$}] (V3) at ($(\VER*0,-\HOR*2.5)$) {};
						\draw[thick,-latex] (V1) to[bend right=0] node[left] {$x$} (V2);
						\draw[thick,-latex] (V2) to[bend right=0] node[left] {$y$} (V3);
		}}}}
		$$
		The application of $a \to H_a$ to the leftmost $a$-labeled hyperedge ``assembles'' the string graph $\SG(xy)$. After the application of this production, attachment nodes of the next hyperedge with selectors from 1 to 6 are glued together; therefore, the application of $b \to H_b$ simply glues all the attachment nodes of the $b$-labeled hyperedge without altering the form of the resulting string graph.
		
		For the sake of comparison, suppose that $bba \in \Language(G)$. Then, we can construct a similar derivation with the only difference that $b \to H_b$ is applied instead of $a \to H_a$ at the first of the three hyperedge-elimination steps. Then, the first external node is glued with the first attachment node of the $y$-labeled edge, the second attachment node of the latter is glued with the first attachment node of the $x$-labeled edge, and the second attachment node of the latter is glued with the second external node. Thus, the resulting hypergraph is $\SG(yx)$ instead of $\SG(xy)$.
		
	\end{example}
	The above explanation underlies the following lemma.
	\begin{lemma}\label{lemma:Gammaprime}
		$\Language(\Gamma^\prime) \subseteq \{\SG(xy),\SG(yx)\}$. Thus $\Gamma^\prime$ is string-generating. Moreover, $\SG(yx) \in \Language(\Gamma^\prime)$ if and only if there is a word starting with $b$ in $\Language(G)$.
	\end{lemma}
\begin{proof}
	Since HRG rule applications can be reordered, we can assume that all rules from $P$ are applied before $a\to H_a$ and $b\to H_b$. Therefore, at this stage, the derived hypergraph has the form $H_0[e/\SG_{[6]}(u)]$ for some $u\in\Language(G)$, where $e$ is the $S$-labeled hyperedge of $H_0$. Recall that $u$ is nonempty. Let $u=c_1\ldots c_n$, and apply the remaining rules to its hyperedges in this order. If $c_1=a$, then the first application glues the nodes with selectors $1,2$, the nodes with selectors $3,4$, and the nodes with selectors $5,6$ at the beginning of $\SG_{[6]}(u)$. This assembles $\SG(xy)$ in the same way as in \cref{example:aba}. If $c_1=b$, the corresponding pairs are $1,4$, $2,5$, and $3,6$, which assembles $\SG(yx)$. In both cases, all six nodes between the first and second hyperedges are glued with the last external node of the resulting string graph; if $n=1$, these are the six nodes at the end of $\SG_{[6]}(u)$. Each subsequent rule application glues the next six nodes with this same node and does not glue any distinct nodes of the string graph already assembled. Thus the final hypergraph is $\SG(xy)$ if $u$ starts with $a$, and $\SG(yx)$ if $u$ starts with $b$. Conversely, a derivation of any such word $u$ gives the corresponding terminal derivation in $\Gamma^\prime$.
\end{proof}
	This concludes the proof.
\end{proof}

This `hypergraph assembling' technique can be generalized in order to show \EXPTIME{}-hardness of a wide variety of decision problems for HRGs. We delay this to the next section. Now, it remains to consider repetition-free HRGs, for which we have already proved the \NP{} upper bound. Again, the lower bound is established for string-generating HRGs.
\begin{lemma}\label{lemma:repetition-free-NP-complete}
	The problem whether a given string-generating repetition-free HRG $\Gamma$ generates a given string graph $H$ is \NP{}-hard.
\end{lemma}
\begin{proof}
	We reduce the 3-exact cover problem: given a collection $\mathcal{C}$ of 3-element subsets of $[3q]$, check if there are pairwise disjoint sets $A_1,\ldots,A_q \in \mathcal{C}$ such that $A_1 \cup \ldots \cup A_q = [3q]$. First, we construct an LCFRS $G$ with nonterminal labels $S,S_1,\ldots,S_q$ such that $\dim(S) = 1$, $\dim(S_m) = 3m$ ($m=1,\ldots,q$), terminal labels $a_1,\ldots,a_{3q},b$ and the following mcf-functions:
	\begin{itemize}
		\item $f_0((x_1,\ldots,x_{3q})) = x_1 b x_2 b \ldots x_{3q} b$,
		\item $f_A^{i_1i_2i_3m}((x_1,\ldots,x_{3m-3})) = (y_1,\ldots, y_{3m})$ where $2 \le m \le q$, $1 \le i_1 < i_2 < i_3 \le 3m$, $A = \{j_1,j_2,j_3\} \in \mathcal{C}$ for $j_1 < j_2 < j_3$, $(y_{i_1},y_{i_2},y_{i_3}) = (a_{j_1},a_{j_2},a_{j_3})$ and $(y_{k_1},\ldots,y_{k_{3m-3}}) = (x_1,\ldots,x_{3m-3})$ where $k_1<\ldots<k_{3m-3}$ is the ordering of the set $[3m] \setminus \{i_1,i_2,i_3\}$,
		\item $f_A() = (a_{j_1},a_{j_2},a_{j_3})$ where $A = \{j_1,j_2,j_3\} \in \mathcal{C}$ for $j_1 < j_2 < j_3$, as above.
	\end{itemize}
	Productions of $G$ are 
	\begin{itemize}
		\item $S \to f_0[S_q]$,
		\item $S_m \to f^{i_1i_2i_3m}_A[S_{m-1}]$ for $2 \le m \le q$, $1 \le i_1 < i_2 < i_3 \le 3m$, and $A \in \mathcal{C}$,
		\item $S_1 \to f_A[]$ for $A \in \mathcal{C}$.
	\end{itemize}
	
	A straightforward induction on $m$ shows that $(w_1,\ldots,w_{3m}) \in \mathrm{yield}_G(S_m)$ iff $w_i \in \{a_1,\ldots,a_{3q}\}$ and
\[
 w_1 \ldots w_{3m} \in \{u_1\} \shuffle \{u_2\} \shuffle \ldots \shuffle \{u_m\}
\]
for some $u_i \in \{a_{j_1}a_{j_2}a_{j_3} \mid j_1<j_2<j_3, \{j_1,j_2,j_3\} \in \mathcal{C}\}$. Here $L_1\shuffle L_2$ is the shuffle of $L_1$ and $L_2$: it consists of words $u_1v_1\cdots u_nv_n$ with $n\ge 0$, $u_i,v_i\in T^\ast$, $u_1\cdots u_n\in L_1$, and $v_1\cdots v_n\in L_2$. Consequently, $a_1 b \ldots a_{3q}b \in \Language(G) = \mathrm{yield}_G(S)$ if and only if there is an exact cover of $[3q]$ by sets from $\mathcal{C}$.
	
	Note that $G$ is repetition-free (its rules are $\varepsilon$-free and interleaving). \cref{proposition:lcfrs-to-hrg} transforms $G$ into a repetition-free HRG $\Gamma$ such that $\Language(\Gamma) = \SG(\Language(G))$, so solving the 3-exact cover problem is reduced to checking whether $\SG(a_1 b \ldots a_{3q}b) \in \Language(\Gamma)$.
\end{proof}

\section{Complexity of Decision Problems for HRGs}\label{section:complexity-decision-problems}

The proof of \EXPTIME{}-hardness of the uniform membership problem for string-generating HRGs turns out to be a very powerful method which can be generalized in order to prove \EXPTIME{}-hardness of a wide variety of decision problems concerning HRGs. In this section, we delineate a class of decision problems to which this method is applicable.

Denote the set of languages generated by HRGs by $\HRL$. Given $H \in \mathcal{H}(C)$ with $C = \{c_1,\ldots,c_n\}$, define its \emph{Parikh vector} $\parikh(H) \eqdef (v_1,\ldots,v_n)$ where $v_i$ is the cardinality of the set $\{e \in E_H \mid \lab_H(e) = c_i\}$. Simply speaking, $\Psi$ counts the number of occurrences of each label in $H$. Using it, we can state the main result.

\begin{theorem}\label{theorem:main-meta}
	Let $\mathcal K \subseteq \HRL$ be some class of languages over a finite label alphabet $C$ and let $L \in \HRL$, $X \in L$, $Y \in \mathcal{H}(C)$ such that
	\begin{enumerate}
		\item $L \notin \mathcal K$,
		\item $L \cup \{Y\} \in \mathcal K$,
		\item $\Psi(X) = \Psi(Y)$,
		\item each of $X$, $Y$ has at least one node.
	\end{enumerate}
	Then, checking whether $\Language(\Gamma) \in \mathcal K$ for an HRG $\Gamma$ is \EXPTIME{}-hard.
\end{theorem}

\begin{proof}[Proof of \cref{theorem:main-meta}]
	Take $L,X,Y$ from the theorem's statement. Since $L \cup \{Y\} \in \mathcal K$ is generated by some HRG, $\type(X) = \type(Y)$, because both $X$ and $Y$ are generated from the handle $S^\bullet$. Fix some node $v_X \in V_X$ and $v_Y \in V_Y$. For technical convenience, assume without loss of generality that $E_X = E_Y = E$ and that $\lab_X(e) = \lab_Y(e)$ for each $e \in E$. This can be achieved by renaming hyperedges because $\Psi(X) = \Psi(Y)$. (Recall that languages generated by HRGs are considered to consist of classes of isomorphic hypergraphs, see \cref{remark:abstract-concrete}.) Also, fix an HRG $\Gamma_0 = \langle N_0,T_0,\Sigma_0,P_0,S_0 \rangle$ such that $\Language(\Gamma_0) = L$.
	
	We reduce the \EXPTIME{}-hard problem from \cref{proposition:EXPTIME-LCFRS}. Similarly to the proof of \cref{theorem:lower_bound_EXPTIME}, we define hypergraphs $H_a$, $H_b$ and $H_0$. Let
	\begin{equation*}
		\UniSetSelsB \eqdef \big(\{0\} \times \type(X) \big) 
		\cup \big(\{1\}\times V_X\big) \cup \big(\{2\} \times V_Y\big)
		\cup \{(3,e,\UniSelA) \mid e \in E, \UniSelA \in \type_X(e)\}.
	\end{equation*}

	\textbf{Definition of $H_a$:}
	\begin{itemize}
		\item $V_{H_a} = V_X$,
		\item $E_{H_a} = \emptyset$,
		\item $\type(H_a) = \Theta \times \{1,2\}$, 
		\begin{itemize}
			\item $\ext_{H_a}(0,\UniSelA,1) = \ext_X(\UniSelA)$ for $\UniSelA \in \type(X)$, 
			\item $\ext_{H_a}(1,v,1) = v$ for $v \in V_X$, 
			\item $\ext_{H_a}(2,v,1) = v_X$ for $v \in V_Y$,
			\item $\ext_{H_a}(3,e,\UniSelA,1) = \att_{X}(e)(\UniSelA)$ for $e \in E$ and $\UniSelA \in \type_X(e)$,
			\item $\ext_{H_a}(\UniSelB,2) = v_X$ for $\UniSelB \in \UniSetSelsB$.
		\end{itemize}
	\end{itemize}
	
	\textbf{Definition of $H_b$:}
	\begin{itemize}
		\item $V_{H_b} = V_Y$,
		\item $E_{H_b} = \emptyset$,
		\item $\type(H_b) = \Theta \times \{1,2\}$, 
		\begin{itemize}
			\item $\ext_{H_b}(0,\UniSelA,1) = \ext_Y(\UniSelA)$ for $\UniSelA \in \type(Y)$, 
			\item $\ext_{H_b}(1,v,1) = v_Y$ for $v \in V_X$, 
			\item $\ext_{H_b}(2,v,1) = v$ for $v \in V_Y$,
			\item $\ext_{H_b}(3,e,\UniSelA,1) = \att_{Y}(e)(\UniSelA)$ for $e \in E$ and $\UniSelA \in \type_Y(e)$,
			\item $\ext_{H_b}(\UniSelB,2) = v_Y$ for $\UniSelB \in \UniSetSelsB$.
		\end{itemize}
	\end{itemize}
	
	\textbf{Definition of $H_0$:}
	\begin{itemize}
		\item $V_{H_0} = \UniSetSelsB \times \{1,2\}$, 
		\item $E_{H_0} = E \cup \{e_0\}$, 
		\item $\type_{H_0}(e) = \type_X(e)$ for $e \in E$,
		\\
		$\type_{H_0}(e_0) = \Theta \times \{1,2\}$,
		\item $\att_{H_0}(e)(\UniSelA) = (3,e,\UniSelA,1)$ for $e \in E$ and $\UniSelA \in \type_X(e)$, 
		\\ $\att_{H_0}(e_0)(\UniSelB,i) = (\UniSelB,i)$ for $(\UniSelB,i) \in \UniSetSelsB \times \{1,2\}$, 
		\item $\lab_{H_0}(e) = \lab_X(e)$ for $e \in E$, $\lab_{H_0}(e_0) = S$,
		\item $\type(H_0)=\type(X)$, $\ext_{H_0}(\UniSelA) = (0,\UniSelA,1)$ for $\UniSelA \in \type(X)$.
	\end{itemize}
	
	The idea behind these hypergraphs is the same as in  the proof of \cref{theorem:lower_bound_EXPTIME}. The start hypergraph $H_0$ contains the hyperedges from $E$, nodes representing every node of $X$ and $Y$, and an additional hyperedge $e_0$. Replacement of $e_0$ by $H_a$ glues $H_0$'s nodes in a way that assembles $X$, and, likewise, replacement of $e_0$ by $H_b$ assembles $Y$.
	
	\begin{example}
		Below we provide an example of $X$, $Y$ and $H_0$ corresponding to those. (The selector names are omitted.)
		\begin{equation*}
			X \;=\; \vcenter{\hbox{{\tikz[baseline=.1ex]{
							\def\W{1}
							\def\H{0.8}
							\node[node,label=left:{\tiny $(1)$}] (U1) at ($(\W*0,\H*0)$) {};
							\node[node] (U2) at ($(\W*1,\H*0)$) {};
							\draw[-latex, thick] (U1) to[bend left=30] node[above] {$c$} (U2);
							\draw[-latex, thick] (U1) to[bend right=30] node[below] {$d$} (U2);
			}}}}
			\qquad
			Y \;=\; \vcenter{\hbox{{\tikz[baseline=.1ex]{
							\def\W{1}
							\def\H{0.8}
							\node[node,label=left:{\tiny $(1)$}] (U2) at ($(\W*0,\H*0)$) {};
							\draw[-latex, thick] (U2) to[out=120,in=60,looseness=30] node[above] {$c$} (U2);
							\draw[-latex, thick] (U2) to[out=-60,in=-120,looseness=30] node[below] {$d$} (U2);
			}}}}
			\qquad
			H_0 \;=\;
			\vcenter{\hbox{{\tikz[baseline=.1ex]{
							\def\HOR{0.4}
							\def\VER{2.1}
							\def\BEND{20}
							\node[hyperedge] (E1) at ($(\VER*0.5,-\HOR*3.5)$) {$S$};
							\foreach \i in {0,1}
							{
								\foreach \j in {0,...,7}
								{
									\pgfmathtruncatemacro\J{\j+1}
									\pgfmathtruncatemacro\JJ{\j+7}
									\node[node, label=left:{\tiny
										\ifnumequal{\i}{0}{\ifnumequal{\j}{0}{$(1)$}}{}
									}] (V\i\j) at ($(\VER*\i,-\HOR*\j)$) {};
								}			
							}
							\foreach \i in {0,1}
							{
								\foreach \j in {0,...,7}
								{
									\pgfmathtruncatemacro\I{\i+1}
									\ifnumgreater{1}{\i}{
										\pgfmathtruncatemacro\J{\j+1}
										\draw[-] (E\I) to[bend \ifnumgreater{\j}{3}{left}{right}=\BEND]  (V\i\j);
									}
									{}
									\ifnumgreater{\i}{0}{
										\pgfmathtruncatemacro\JJ{\j+7}
										\draw[-] (E\i) to[bend \ifnumgreater{\j}{3}{right}{left}=\BEND] (V\i\j);
									}
									{}
								}
							}
							\draw[thick,-latex] (V04) to[bend right=0] node[left] {$c$} (V05);
							\draw[thick,-latex] (V06) to[bend right=0] node[left] {$d$} (V07);
			}}}}
		\end{equation*}
	\end{example}
	
	Given an LCFRS $G$ over the terminal alphabet $\{a,b\}$ not generating the empty word, apply \cref{proposition:lcfrs-to-hrg-generalised}, which outputs an HRG $\Gamma_1 = \langle N_1,\{a,b\},\Sigma_1,P_1,S_1 \rangle$ such that $\Language(\Gamma_1) = \SG_{\UniSetSelsB}(\Language(G))$. Since $\UniSetSelsB$-string hypergraphs have the type $\UniSetSelsB \times\{1,2\}$, $\type(S_1) = \UniSetSelsB \times\{1,2\}$ as well. Therefore, we can assume that $S_1 = S$, the label that occurs in $H_0$. 
	
	Without loss of generality, all sets $N_0$, $T_0$, $N_1$, $\{a,b\}$ are disjoint. Let $P_2 = \{S^\prime \to S_0^\bullet, S^\prime \to H_0, a \to H_a, b \to H_b\}$, where $S^\prime$ is a fresh nonterminal of type $\type(X)$. Define
	\[
	\Gamma = \langle N_0 \cup N_1\cup \{a,b, S^\prime\}, T_0, \Sigma_0 \cup \Sigma_1, P_0 \cup P_1 \cup P_2, S^\prime \rangle.
	\] 
	
	Starting with $(S^\prime)^\bullet$, one can either apply $S^\prime \to S_0^\bullet$ or $S^\prime \to H_0$. In the first case, the derivation can be prolonged using only rules of $\Gamma_0$, therefore, any resulting terminal hypergraph is from $\Language(\Gamma_0) = L$. In the second case, a derivation of a terminal hypergraph can be divided into two stages:
	\begin{enumerate}
		\item applications of rules from $P_1$;
		\item applications of the rules $a \to H_a$, $b \to H_b$.
	\end{enumerate}
	Applications of rules from $P_1$ eventually transform $H_0$ into a hypergraph $H^\prime$, whose nonterminal labels belong to $\{a,b\}$, because other nonterminal labels cannot be eliminated at the second stage. Thus, the derivation of $H^\prime$ from $H_0$ corresponds to a derivation of some terminal string $u \in \{a,b\}^\ast$ in the LCFRS $G$, and $H^\prime = H_0[e_0/\SG_{\UniSetSelsB}(u)]$. Denote the hyperedges of $\SG_{\UniSetSelsB}(u)$ by $s_1,\ldots,s_n$, as in \cref{definition:generalized-string-graph}. Note that $n>0$ because $G$ does not generate the empty word. 
	
	Since HRG rule applications can be reordered, we can assume without loss of generality that rules at the second stage are applied to $s_1,\ldots,s_n$ in this order. If $u$ starts with $a$, then $\lab_{H^\prime}(s_1) = a$, and the application of $a \to H_a$ ``assembles'' the hypergraph $X$ (in the same manner as in \cref{example:aba}). The remaining rule applications do not alter the connections between the hyperedges and nodes of $X$, so ultimately one obtains $X$. Note that $X$ already belongs to $L$. Similarly, if $u$ starts with $b$, then the resulting terminal hypergraph is $Y$. Thus, 
	\[
	\Language(\Gamma)
	=
	\begin{cases}
		L \cup \{Y\} & \exists u \in \Language(G) \  u = bu^\prime \\
		L & \text{else}
	\end{cases}
	\]
	This provides the desired reduction: given $G$, we have constructed $\Gamma$ in polynomial time such that $\Language(\Gamma) \in \mathcal K$ $\Longleftrightarrow$ there is a word starting with $b$ in $\Language(G)$.
\end{proof}

One example of this theorem's instantiation is
\begin{corollary}
	For a fixed $k \ge 2$, checking whether an HRG generates at least $k$ (non-isomorphic) hypergraphs is \EXPTIME{}-hard.
\end{corollary}
\begin{proof}
	Take $k$ edgeless hypergraphs $X_1,\ldots,X_k$ of type $\emptyset$, with respectively $1,\ldots,k$ isolated nodes. They are non-isomorphic and have the same Parikh vector. Let $L = \{X_1,\ldots,X_{k-1}\}$, let $X = X_1$ and $Y = X_k$. $L,X,Y$ satisfy all properties of \cref{theorem:main-meta}.
\end{proof}

\begin{remark}\label{remark:nonemptiness-P}
	Checking whether the language of a given HRG is non-empty can be done in polynomial time. Indeed, $\Language(\Gamma) \ne \emptyset$ iff $\Language(G) \ne \emptyset$ for any linearisation $G$ of $\Gamma$, and the non-emptiness problem for context-free languages is in \textup{P}. Note that the class $\mathcal K$ consisting of all non-empty languages does not satisfy the conditions of \cref{theorem:main-meta}. In fact, deciding any property that holds for an HRG iff it holds for any of its linearisations and which can be tested in polynomial time for context-free grammars is in \textup{PTIME}, which of course is in line with \cref{theorem:main-meta}.
\end{remark}

\subsection{Non-Parikh Properties}

Below, we derive a general corollary of \cref{theorem:main-meta}, showing that it is \EXPTIME{}-hard to check whether an HRG generates at least one graph satisfying a given property, assuming this property does not rely on the Parikh vector of a hypergraph only.
\begin{definition}[Non-Parikh property]
	Let $\mathrm{Grph}_\ast=\mathcal H(\{\ast\})$ denote the set of graphs over the one-letter alphabet $\{\ast\}$ such that $\type(\ast)=\{1,2\}$. A \emph{non-Parikh graph property} is a subset $\mathcal G\subseteq\mathrm{Grph}_\ast$ such that there are two graphs $X,Y\in\mathrm{Grph}_\ast$ with $\type(X)=\type(Y)$, with the same number of edges, and with at least one node each, such that $X\in\mathcal G$ while $Y\notin\mathcal G$.
	
	Given a hypergraph $H$, let $\unl(H)$ denote the hypergraph $\langle V_H,\allowbreak E_H,\allowbreak\lab^\prime,\allowbreak\att_H,\allowbreak\ext_H\rangle$ where $\lab^\prime(e)=\ast$ if $\type_H(e)=\{1,2\}$ and $\lab^\prime(e)=\lab_H(e)$ otherwise. Thus $\unl$ erases edge labels, placing the blank label instead, and leaves all other hyperedge labels unchanged. For a language $L$, let $\unl(L)=\{\unl(H)\mid H\in L\}$.
\end{definition}

\begin{corollary}\label{corollary:non-Parikh}
	Let $\mathcal G$ be a non-Parikh graph property. Given an HRG $\Gamma$, checking whether $\unl(\Language(\Gamma)) \subseteq \mathcal G$ is \EXPTIME{}-hard, as well as whether $\unl(\Language(\Gamma)) \cap \mathcal G \ne \emptyset$. 
\end{corollary}
\begin{proof}
	Take a graph $X \in \mathcal G$ and a graph $Y \in \mathrm{Grph}_\ast \setminus \mathcal G$, with at least one node each, such that $\type(X)=\type(Y)$ and $\Psi(X)=\Psi(Y)$. Let $L = \{X\}$. Let $\mathcal K = \{\Language(\Gamma) \mid \unl(\Language(\Gamma)) \not\subseteq \mathcal G\}$. Then $L \in \HRL\setminus \mathcal K$ and $L \cup\{Y\} \in \mathcal K$. It remains to apply \cref{theorem:main-meta}. Finally, mind that the complement of an \EXPTIME{}-hard problem is also \EXPTIME{}-hard.
	
	For the second kind of decision problems, take $\mathcal K = \{\Language(\Gamma) \mid \unl(\Language(\Gamma)) \cap \mathcal G \ne \emptyset\}$ and switch $X$ and $Y$.
\end{proof}

In particular, checking that an HRG generates only\slash some string\slash connected\slash Eulerian\slash Hamiltonian\slash acyclic\slash planar graphs is \EXPTIME{}-hard.

\begin{remark}\label{remark:one-letter-HRG}
	In \cref{remark:one-letter-sg-HRG}, we noticed that the uniform membership problem for string-generating HRGs over a one-letter alphabet is in P. This is not the case for general HRGs over a one-letter alphabet. Consider the graph property $\mathcal G = \{\SG(\varepsilon)\}$, i.e.~the property of being equal to the string graph $\vcenter{\hbox{{\tikz[baseline=.1ex]{
					\node[node, label=left:{\tiny $(1)$}, label=right:{\tiny $(2)$}] (V1) at ($(0,0)$) {};
	}}}}$ representing the empty word. \cref{corollary:non-Parikh} implies that checking whether $\SG(\varepsilon)$ is generated by an HRG is \EXPTIME{}-complete, similarly to DTWTs and MCFGs. 
\end{remark}

\subsection{What About Repetition-Free Grammars?}

The main construction used in the proofs of \cref{theorem:lower_bound_EXPTIME,theorem:main-meta} relies heavily on repetitions of attachment nodes and external nodes. One might be curious about whether a metatheorem similar to \cref{theorem:main-meta,corollary:non-Parikh} holds for repetition-free HRGs. We already know that this is not the case in general. E.g.~the uniform membership problem for repetition-free HRGs is only \NP{}-complete, not \EXPTIME{}-hard (under the $\NP{}\ne\EXPTIME{}$ assumption). Nevertheless, there is still a way to transfer \cref{theorem:main-meta} to repetition-free HRGs. We fix a label alphabet $C$.

\begin{definition}
	A \emph{patch family} is a sequence of hypergraphs $\Patch = (\Patch_n)_{n \ge 1}$ where $\Patch_n \in \mathcal H(C)$ has the type $\type(\Patch_n) = [n]$ and $\ext_{\Patch_n}$ is injective; moreover, the function $n \mapsto \Patch_n$ is computable in polynomial time.
	
	Let $O_n$, for $n \ge 1$, be a hypergraph such that $V_{O_n} = \{v\}$, $E_{O_n} = \emptyset$, $\type(O_n) = [n]$ and $\ext_{O_n}(i) = v$ for $i=1,\ldots,n$.
	
	We write $G \rightsquigarrow_{\Patch} H$ for hypergraphs $G, H \in \mathcal H(C)$ in case there is a hypergraph $G^\prime$ over $C\cup\{d_n\}$ and a hyperedge $e\in E_{G^\prime}$, where $d_n$ is a fresh label of type $[n]$ and $e$ is the only hyperedge labeled by $d_n$ in $G^\prime$, such that $\type(e) = [n]$, $G = G^\prime[e/\Patch_n]$, $H = G^\prime[e/O_n]$.
\end{definition}

Informally, $H$ is obtained from $G$ by compressing a subhypergraph of $G$ isomorphic to $\Patch_n$ into a single node. Some external nodes of this subhypergraph may already coincide, since $\att_{G^\prime}(e)$ is not required to be injective.

\begin{definition}
	Let $\equiv_{\Patch}$ be the smallest equivalence relation that contains $\rightsquigarrow_{\Patch}$. For two languages $L_1, L_2 \subseteq \mathcal H(C)$, we write $L_1 \equiv_{\Patch} L_2$ in case $\{[H]_{\equiv_{\Patch}} \mid H \in L_1\} = \{[H]_{\equiv_{\Patch}} \mid H \in L_2\}$.
	
	We say that a class of languages $\mathcal K$ over a finite label alphabet $C$ is \emph{$\Patch$-invariant} if $L_1 \in \mathcal K$ is equivalent to $L_2 \in \mathcal K$ for every two languages $L_1,L_2\in\HRL$ over $C$ such that $L_1 \equiv_{\Patch} L_2$.
\end{definition}

Using these definitions, we can state a metatheorem for repetition-free HRGs. In fact, let us restrict the class of HRGs even further and consider \emph{apex HRGs}.

\begin{definition}
	An HRG is \emph{apex} if, for each of its productions $X \to D$, there is no hyperedge in $D$ labeled by a nonterminal symbol and incident to an external node of $D$.
\end{definition}

\begin{theorem}\label{theorem:main-meta-repetition-free}
	Let $\mathcal K \subseteq \HRL$ be a $\Patch$-invariant class of languages over a finite label alphabet $C$. Let $L \in \HRL$, $X \in L$, $Y \in \mathcal{H}(C)$ such that
	\begin{enumerate}
		\item $L \notin \mathcal K$,
		\item $L \cup \{Y\} \in \mathcal K$,
		\item $\Psi(X) = \Psi(Y)$,
		\item each of $X$, $Y$ has at least one node.
	\end{enumerate}
	Then, checking whether $\Language(\Gamma) \in \mathcal K$ for an apex repetition-free HRG $\Gamma$ is \EXPTIME{}-hard.
\end{theorem}
The only difference in the premise of this theorem as compared to \cref{theorem:main-meta} is that $\mathcal K$ is $\Patch$-invariant for some appropriate patch family $\Patch$.

The above theorem follows from the following lemma.

\begin{lemma}\label{lemma:polytime-repetition-free}
	For each fixed patch family $\Patch$, there exists a polynomial-time algorithm that, for every HRG $\Gamma$ over the terminal alphabet $C$, produces an apex repetition-free HRG $\Gamma^\prime$ over $C$ such that $\Language(\Gamma) \equiv_{\Patch} \Language(\Gamma^\prime)$.
\end{lemma}

Let us prove the theorem first.

\begin{proof}[Proof of \cref{theorem:main-meta-repetition-free}]
	We know that the problem of checking whether $\Language(\Gamma) \in \mathcal K$ is \EXPTIME{}-hard for arbitrary $\Gamma$. We reduce the latter problem to the one stated in \cref{theorem:main-meta-repetition-free}. Given an arbitrary HRG $\Gamma$, we apply a polynomial-time transformation from \cref{lemma:polytime-repetition-free} and obtain an apex repetition-free HRG $\Gamma^\prime$ such that $\Language(\Gamma) \equiv_{\Patch} \Language(\Gamma^\prime)$. Then, by invariance of $\mathcal K$ under $\equiv_{\Patch}$, the membership $\Language(\Gamma) \in \mathcal K$ is equivalent to the membership $\Language(\Gamma^\prime) \in \mathcal K$. This provides a desired reduction.
\end{proof}

\begin{proof}[Proof of \cref{lemma:polytime-repetition-free}]
	Let $\Gamma=\langle N,C,\Sigma,P,S\rangle$. For each production $A\to D$, we construct a hypergraph $D^\prime$ as follows. For every node $v\in V_D$, put
	\(
	I_v=\{s\in\type(D)\mid\ext_D(s)=v\}\) and
	\(J_v=\{(e,s)\mid e\in E_D,\ \lab_D(e)\in N,\
	s\in\type(e),\ \att_D(e)(s)=v\}\). The set $I_v$ contains information about all the selectors mapped to $v$ by the external node function; $J_v$ collects the pairs consisting of a nonterminal hyperedge and a selector mapped to $v$ by its attachment function.
	For each $s\in I_v$, introduce a fresh node $v_s^{\mathrm{ext}}$ and redefine
	$\ext_D(s)=v_s^{\mathrm{ext}}$. For each $(e,s)\in J_v$, introduce a fresh node
	$v_{e,s}^{\mathrm{att}}$ and redefine
	$\att_D(e)(s)=v_{e,s}^{\mathrm{att}}$.
	Retain the original node $v$ and leave the attachment functions of terminal
	hyperedges unchanged.
	
	If $m_v=|I_v|+|J_v|>0$, introduce a new hyperedge $e_v$ such that $\type(e_v) = [m_v+1]$, $\lab(e_v) = d_{m_v+1}$, its first attachment node is $v$ and its other attachment nodes are all the new nodes introduced for $v$ (namely, $v_s^{\mathrm{ext}}$ and $v_{e,s}^{\mathrm{att}}$). The hypergraph obtained from $D$ in this way is denoted by $\hat D$. Finally, let $D^\prime$ be obtained from $\hat D$ by replacing, for each $v$, the hyperedge $e_v$ by $\Patch_{m_v+1}$.
	Define a desired HRG $\Gamma^\prime$ by replacing every production $A\to D$ by $A\to D^\prime$.
	
	The resulting grammar $\Gamma^\prime$ is repetition-free: distinct external
	selectors in productions are mapped to distinct nodes $v_s^{\mathrm{ext}}$, and distinct selectors of
	every nonterminal hyperedge are also mapped to distinct fresh nodes $v_{e,s}^{\mathrm{att}}$.
	Moreover, no attachment node of a nonterminal hyperedge is an external
	node of $D^\prime$, since the nodes $v_{e,s}^{\mathrm{att}}$ and
	$v_s^{\mathrm{ext}}$ are distinct. Thus $\Gamma^\prime$ is apex.
	
	Finally, observe that $D^\prime \rightsquigarrow_{\Patch}^\ast D$ because $D$ is obtained from $\hat D$ by replacing each $e_v$ by $O_{m_v+1}$. Therefore,
	$\Language(\Gamma)\equiv_{\Patch}\Language(\Gamma^\prime)$, because hyperedge replacement preserves $\rightsquigarrow_{\Patch}$.
\end{proof}

Let us show some applications of \cref{theorem:main-meta-repetition-free}.

\begin{corollary}
	It is \EXPTIME{}-hard to check whether an apex repetition-free HRG generates some/only weakly connected graphs.
\end{corollary}

\begin{proof}
	Consider the label alphabet $C = \{\ast\}$ with $\type(\ast) = [2]$. Let $L\in\mathcal K^{\exists\mathit{con}}$ iff $L$ contains a weakly connected graph. Then $\mathcal K^{\exists\mathit{con}}$ is $\Patch^{\mathit{comp}}$-invariant where $\Patch^{\mathit{comp}}_n$ is the graph on nodes $v_1,\ldots,v_n$, with $\ext(i)=v_i$ and one edge from $v_i$ to $v_j$ whenever $i \ne j$. Indeed, $\Patch^{\mathit{comp}}_n$ is the complete graph, and compressing it into a single node preserves connectedness/disconnectedness of a graph.
	
	Similarly, let $L\in\mathcal K^{\forall\mathit{con}}$ iff $L$ contains only weakly connected graphs. Then $\mathcal K^{\forall\mathit{con}}$ is $\Patch^{\mathit{comp}}$-invariant for the same reason.
\end{proof}

This corollary reproves the same \EXPTIME{}-hardness result by Wanke \cite{Wanke94}, who proved it for apex NLC grammars, which can also be defined as apex repetition-free HRG grammars. However, \cref{theorem:main-meta-repetition-free} can be applied to other properties too.

\begin{corollary}\label{corollary:colorable}
	For each fixed $k \ge 2$, it is \EXPTIME{}-hard to check whether an apex repetition-free HRG generates some/only $k$-colorable graphs.
\end{corollary}

To recall, a graph $H$ is $k$-colorable if there is a mapping $\mathit{col}:V_H \to [k]$ such that, whenever $v_1,v_2 \in V_H$ are adjacent, $\mathit{col}(v_1) \ne \mathit{col}(v_2)$.

\begin{proof}
	Let $L\in\mathcal K^{\exists\mathit{col}_k}$ iff $L$ contains a $k$-colorable graph. Let $\Patch^{\mathit{col}_k}_n$ have $n$ distinct external nodes and $k-1$ internal nodes. Join every two distinct internal nodes by an edge and join every external node to every internal node, orienting the edges arbitrarily. 
	
	Assume that $H_1\rightsquigarrow_{\Patch^{\mathit{col}_k}}H_2$, i.e.~$H_1=H[e/\Patch^{\mathit{col}_k}_n]$ and $H_2=H[e/O_n]$ for some $H,e,n$. Assume that $H_1$ is $k$-colorable. Since each pair of internal nodes in \(\Patch^{\mathit{col}_k}_n\) is adjacent, they must be coloured using $k-1$ different colours. Therefore, the attachment nodes of $e$ are coloured using the remaining $k$-th colour $c$. This allows one to colour $H_2$ by assigning $c$ to the node obtained after replacing $e$ by $O_n$.
	
	Conversely, let $H_2$ be $k$-colorable and let $c$ be the colour of the node obtained after replacing $e$ by $O_n$. Give every attachment node of $e$ colour $c$, retain the colours of all other nodes of $H$, and colour $H_1$ by assigning $k-1$ colours distinct from $c$ to the internal nodes of $\Patch^{\mathit{col}_k}_n$.
\end{proof}

\subsection{Time-Bounded Compatible Properties}

Let us now discuss for which decision problems the \EXPTIME{} lower bound established by \cref{theorem:main-meta} is tight. A good starting point is Filter Theorem \cite{Habel92,DrewesKH97}, which proves decidability of checking whether an HRG satisfies a given \emph{compatible property}. Below, we introduce a time-bounded version of compatible properties.

\begin{definition}[$t(n)$-compatible property]\label{definition:compatible}
	Fix an increasing function $t:\Nat \to \Nat$, a selector set $\Sigma_0$, and two disjoint $\Sigma_0$-typed alphabets $N_0,T_0$. Given a hypergraph $H$, we shall denote the set $\{e \in E_H \mid \lab_H(e) \in N_0\}$ of its nonterminal-labeled hyperedges by $E^N_H$.
	
	The following are dramatis personae of the definition.
	\begin{itemize}
		\item $\UniClassHRG \subseteq \HRG$ is some subclass of HRGs such that
		if $\langle N,T,\Sigma, P,S\rangle \in \UniClassHRG$, then $N \subseteq N_0$, $T \subseteq T_0$, $\Sigma \subseteq \Sigma_0$.
		\item $I$ is a countable set of indices, which we identify with $\Nat$ until the end of the definition. 
		\item $\UniProp$ is a predicate defined on pairs $(H,i)$ where $H \in \mathcal{H}(T_0)$ and $i \in I$.
		\item $\UniProp^\prime$ is a predicate defined on triples $(H,\ass,i)$ where $H \in \mathcal{H}(N_0 \cup T_0)$, $\ass: E^N_H \to I$ is a mapping from nonterminal hyperedges to indices, and $i \in I$. We require that 
		\begin{itemize}
			\item if $\UniProp^\prime(H,\ass,i)$ is true, then $i < t(\vert H \vert)$ and $\ass(e) < t(\vert H \vert)$ for each $e \in E^N_H$ and
			\item $\UniProp^\prime \in \DTIME(t(n))$.
		\end{itemize}
		\item $I_0$ is a computable function mapping finite subsets of $\Sigma_0$ to finite subsets of $I$ such that $I_0(X)$ is computed in time $\le t(\ell(X))$ for finite $X \subseteq \Sigma_0$, where $\ell(X)$ is the length of the encoding of $X$.
		\item $\UniProp_0 \subseteq \mathcal{H}(T_0)$ is a predicate such that $\UniProp_0(H)$ is true iff so is $\UniProp(H,i)$ for some $i \in I_0(\type(H))$.
	\end{itemize}
	Finally, the main requirement is: for each HRG $\Gamma = \langle N,T,\Sigma,P,S\rangle \in \UniClassHRG$, for each derivation $A^\bullet \Rightarrow R \Rightarrow^\ast H$ using this grammar where $H \in \mathcal H(T)$, and for each $i \in I$, $\UniProp(H,i)$ holds iff there is a mapping $\ass:E^N_R \to I$ such that $\UniProp^\prime(R,\ass,i)$ holds and $\UniProp(H(e),\ass(e))$ holds for each $e \in E^N_R$. Here $H(e)$ is the hypergraph obtained from the hyperedge $e$ in the derivation $R \Rightarrow^\ast H$. (Consequently, $H$ is obtained from $R$ by replacing each $e \in E^N_R$ by $H(e)$.)
	
	If all the above is satisfied, the property $\UniProp_0$ is called \emph{$t(n)$-compatible (in the class $\UniClassHRG$)}. A property is \emph{\EXPTIME{}-compatible} if it is $2^{p(n)}$-compatible for some polynomial $p(n)$.
\end{definition}
\begin{theorem}\label{theorem:bounded-filter-theorem}
	For any \EXPTIME{}-compatible property, deciding whether a given HRG $\Gamma \in \UniClassHRG$ generates some hypergraph satisfying this property is in \EXPTIME{}.
\end{theorem}

\begin{proof}
	The proof follows the lines of the proof of the Filter Theorem, see \cite[Theorem 2.6.2]{DrewesKH97}. Let $t(n)=2^{p(n)}$ for some increasing polynomial $p(n)$; the notation of Definition \ref{definition:compatible} is used throughout. Given $\Gamma=\langle N,T,\Sigma,P,S\rangle\in\UniClassHRG$, define $\Gamma^\prime=\langle N^\prime,T,\Sigma,P^\prime,S^\prime\rangle$ as follows:
	\begin{itemize}
		\item $N^\prime=(N\times\{i\in I\mid i<t(|\Gamma|)\})\cup\{S^\prime\}$, where $S^\prime$ is fresh, $\type((A,i))=\type(A)$, and $\type(S^\prime)=\type(S)$;
		\item $P^\prime$ consists of rules
		\begin{itemize}
			\item $(A,i)\to(R,\ass)$ such that $(A\to R)\in P$ and $\UniProp^\prime(R,\ass,i)$ holds, where $(R,\ass)=\langle V_R,E_R,\lab,\att_R,\ext_R\rangle$ with $\lab(e)=\lab_R(e)$ if $\lab_R(e)\in T$ and $\lab(e)=(\lab_R(e),\ass(e))$ otherwise;
			\item $S^\prime\to(S,i)^\bullet$ for $i\in I_0(\type(S))$ such that $i<t(|\Gamma|)$.
		\end{itemize}
	\end{itemize}
	When constructing $P^\prime$, we should only consider $i<t(|R|)$ and $\ass(e)<t(|R|)$, because otherwise $\UniProp^\prime(R,\ass,i)$ is false, by Definition \ref{definition:compatible}. The number of such functions $\ass$, for fixed $(A\to R)\in P$, is at most $t(|R|)^{|E_R^N|}$. Including the choice of $i$ gives at most $t(|R|)^{|E_R^N|+1}$ possibilities. The number of rules of the form $S^\prime\to(S,i)^\bullet$ does not exceed the time required for computing $I_0(\type(S))$, which in turn is at most $t(|\Gamma|)$. Thus the total number of rules in $P^\prime$ is bounded by
	\[
	t(|\Gamma|)^{|\Gamma|+1}|P|+t(|\Gamma|),
	\]
	and hence the size of $\Gamma^\prime$ is at most $2^{q(|\Gamma|)}$ for some polynomial $q(n)$. Besides, checking that a rule $(A,i)\to(R,\ass)$ is valid, i.e.~$\UniProp^\prime(R,\ass,i)$ is true, is done in exponential time w.r.t. the size of $(R,\ass,i)$. Since $i<t(|R|)=2^{p(|R|)}$, the binary encoding of $i$ has length at most $p(|R|)+1$; the same holds for each $\ass(e)$. Thus the total time required to build $P^\prime$ and verify correctness of its rules is exponential.

	A straightforward induction on derivation trees shows that, for $A\in N$ and $i<t(|\Gamma|)$, the hypergraphs derived from $(A,i)^\bullet$ are precisely those derived from $A^\bullet$ that satisfy $\UniProp(H,i)$. In one direction, the indices assigned to the nonterminal hyperedges in a rule of $P^\prime$ satisfy $\UniProp$ by induction, so compatibility gives $\UniProp(H,i)$. In the other direction, compatibility supplies these indices, and the induction hypothesis supplies the corresponding derivations. Consequently,
	\[
	\Language(\Gamma^\prime)=\{H\in\Language(\Gamma)\mid\UniProp_0(H)\}.
	\]
	Therefore, checking the property of interest is the same as checking nonemptiness of $\Language(\Gamma^\prime)$. This is done in polynomial time in $|\Gamma^\prime|$, as for context-free grammars (\cref{remark:nonemptiness-P}).
\end{proof}

Many compatible properties, in the sense of \cite{DrewesKH97, Habel92}, can be shown to be \EXPTIME{}-compatible, which yields a tight EXPTIME bound on the complexity of corresponding decision problems, assuming those properties are non-Parikh as well. We provide one particular example.

\begin{proposition}
	Being/non-being a string graph is an \EXPTIME{}-compatible property in the class $\HRG$.
\end{proposition}

\begin{proof}
	Fix $N_0,T_0,\Sigma_0$, as in Definition
	\ref{definition:compatible}. Below, we introduce some
	notation and terminology.
	\begin{itemize}
		\item Let $\ast$ and $\unl$ be as defined above. A graph $H\in\mathcal H(\{\ast\})$ is called \emph{unlabeled}.
		
		\item Let $\pi \eqdef (\ast \to \SG(\ast\ast))$
		be a production that replaces a single edge with
		two consecutive edges.
		
		\item A hypergraph $H^\prime$ is called
		\emph{minimal below $H$} if
		$H^\prime \Rightarrow^\ast_\pi H$ and there is no
		$H^{\prime\prime}$ such that
		$H^{\prime\prime} \Rightarrow_\pi H^\prime$. 
		
		\item Let us call a node $v$ in an unlabeled graph
		$H$ \emph{bad} if it is not external
		($v \notin \ran(\ext_H)$), or its in-degree is
		greater than $1$, or its out-degree is greater
		than $1$.
		
		\item A \emph{good graph} is an unlabeled graph
		with neither bad nodes nor directed cycles.
	\end{itemize}
	Informally, a good graph consists of several string-like
	subgraphs, possibly including isolated nodes, and all its
	nodes are external. In particular, a good graph of type
	$X$ has at most $\vert X\vert$ nodes and at most
	$\vert X\vert$ edges.

	Let us prove the following lemma.
	\begin{lemma}
		For every hypergraph $H$, there exists a unique hypergraph, up to isomorphism, which is minimal below $H$. This hypergraph, denoted by $\red(H)$, can be computed from $H$ in polynomial time.
	\end{lemma}
	\begin{proof}
		By Newman's lemma \cite{Newman42}, it suffices to show that the inverse of $\Rightarrow_\pi$ is terminating and locally confluent, where hypergraphs are considered up to isomorphism. This relation is terminating, i.e.~there is no infinite sequence of the form $\ldots \Rightarrow_\pi H_1 \Rightarrow_\pi H_0$. Indeed, $\Rightarrow_\pi$ increases the number of nodes in a hypergraph. Let us explain why it is locally confluent, that is, $H_1\Rightarrow_\pi H$ and $H_2\Rightarrow_\pi H$ imply the existence of $H^\prime$ such that $H^\prime\Rightarrow_\pi^\ast H_1$ and $H^\prime\Rightarrow_\pi^\ast H_2$.

		Note that $G^\prime\Rightarrow_\pi G$ is equivalent to the fact that $G^\prime$ is obtained from $G$ by
		\begin{enumerate}
			\item choosing a nonexternal node $v\in V_G$ and two distinct $\ast$-labeled edges $e_1,e_2\in E_G$ such that $\att_G(e_1)(2)=\att_G(e_2)(1)=v$, and no other pair $(e,s)$ satisfies $\att_G(e)(s)=v$;
			\item removing $v,e_1,e_2$ from $G$;
			\item adding a new $\ast$-labeled edge $e_0$ with
			\[
			\begin{aligned}
			\att_{G^\prime}(e_0)(1)&=\att_G(e_1)(1),\\
			\att_{G^\prime}(e_0)(2)&=\att_G(e_2)(2).
			\end{aligned}
			\]
		\end{enumerate}
		In this case, let us call the triple $(\att_G(e_1)(1),v,\att_G(e_2)(2))$ \emph{active for $G^\prime\Rightarrow_\pi G$}.

		Suppose that $H_1\Rightarrow_\pi H$ and $H_2\Rightarrow_\pi H$, and let $(u_1,v_1,w_1)$ and $(u_2,v_2,w_2)$ be active for these applications, respectively. If $v_1=v_2$, the two applications remove the same node and edges, so $H_1=H_2$ up to isomorphism. Otherwise, the following cases arise.
		\begin{enumerate}
			\item $v_1=u_2$. This is equivalent to $w_1=v_2$, because there is only one edge from $v_1$ and only one edge to $v_2$. Two subcases arise.
			\begin{enumerate}
				\item $u_1=w_1$. Then $u_1,v_1$ form a directed cycle of length two. These nodes are nonexternal and have no other incident hyperedges. Both $H_1$ and $H_2$ are obtained from $H$ by replacing this cycle by one node with a loop. Hence they are isomorphic, and we can take $H^\prime=H_1$.
				\item $u_1\ne w_1$. The three consecutive edges from $u_1$ through $v_1,v_2$ to $w_2$ can be replaced by one edge in either order: after removing $v_1$, we can remove $v_2$, and conversely. Let $H^\prime$ be the hypergraph obtained after both removals. Then $H^\prime\Rightarrow_\pi H_1$ and $H^\prime\Rightarrow_\pi H_2$.
			\end{enumerate}
			\item $v_2=u_1$. This is similar to the previous case.
			\item Neither $v_1=u_2$ nor $v_2=u_1$. Then the two applications remove disjoint sets of nodes and edges. Applying both inverse transformations, in either order, gives the required $H^\prime$.
		\end{enumerate}
		Repeatedly applying the inverse to $\Rightarrow_\pi$ while possible, starting from $H$, computes $\red(H)$ in polynomial time.
	\end{proof}

	For a hypergraph $H$, define $s(H)$ as
	\[
	s(H) \eqdef
	\begin{cases}
		\red(\unl(H)),
		& \text{if $\red(\unl(H))$ is good},\\
		\mathit{bad},
		& \text{otherwise},
	\end{cases}
	\]
	where $\mathit{bad}$ is a fresh symbol. 
	
	Define 
	\begin{itemize}
		\item the set of indices $I$ as $\{H \in \mathcal{H}(\{\ast\}) \mid H ~\text{is good} \} \cup \{\mathit{bad}\}$ where $\mathit{bad}$ is a fresh symbol,
		\item $\UniProp(H,i)$ to be true if $i = s(H)$,
		\item $\UniProp^\prime(H,\ass,i)$ to be true if $i$ is either $\mathit{bad}$ or is a good graph of type $\type(H)$, each $\ass(e)$ is either $\mathit{bad}$ or is a good graph of type $\type(e)$, and one of the two holds:
		\begin{itemize}
			\item $\ass(e) = \mathit{bad}$ for some $e \in E^N_H$ and $i = \mathit{bad}$, or
			\item $\ass(e)$ is good for each $e \in E^N_H$, and $s(H[\ass])=i$. (Here $H[\ass]$ is obtained from $H$ by replacing $e$ by $\ass(e)$ for each $e \in E^N_H$.)
		\end{itemize}
		\item $I_0(X)\eqdef\{\SG(\varepsilon),\SG(\ast)\}$ if $X=\{1,2\}$, and $I_0(X)\eqdef\emptyset$ otherwise.
	\end{itemize}
	
	Let us verify the main requirement of Definition \ref{definition:compatible}. First, if $s(K)=\mathit{bad}$, then replacing a hyperedge by $K$ in any hypergraph cannot yield a terminal hypergraph $H$ with $s(H)$ good. Indeed, if $\unl(K)$ has a hyperedge of type other than $\{1,2\}$, a node of in-degree or out-degree greater than one, or a directed cycle, the same property remains after replacement and reduction. Replacement may identify nodes but does not identify edges, so in/out-degrees of nodes cannot decrease and a directed cycle remains a directed closed walk containing a directed cycle. 
	
	Next, if $s(K)$ is good, then $s(K)\Rightarrow_\pi^\ast\unl(K)$. This derivation only subdivides edges by fresh nonexternal nodes. It can therefore be performed after $s(K)$ is inserted into any hypergraph, even if some external nodes are identified during insertion. Consequently, replacing a hyperedge by $K$ or by $s(K)$ gives the same reduced hypergraph after labels are erased.

	Now consider $A^\bullet\Rightarrow R\Rightarrow^\ast H$ as in Definition \ref{definition:compatible}, and put $\ass(e)=s(H(e))$. If some $\ass(e)=\mathit{bad}$, the first observation gives $s(H)=\mathit{bad}$. Otherwise, the second observation and uniqueness of the reduced hypergraph give $s(H)=s(R[\ass])$. Thus $\UniProp(H,i)$ holds exactly when $\UniProp^\prime(R,\ass,i)$ holds. Moreover, $\UniProp(H(e),i)$ is true for exactly one $i$, which proves both directions of the required equivalence.

	In case $\ass(e)$ is good for each $e \in E^N_H$, verifying $\UniProp^\prime(H,\ass,i)$ is done deterministically in exponential time: one replaces $e$ by $\ass(e)$ for each $e \in E^N_H$ (this is done in polynomial time), then applies the inverse of the rule $\pi$ while possible (this is polynomial time as well), and finally checks whether the reduction is good and compares its index with $i$. Thus all the bounds in Definition \ref{definition:compatible} hold for a sufficiently large polynomial $p$ and $t(n)=2^{p(n)}$.

	Finally, according to the definition, $\UniProp_0(H)$ is true iff $\type(H)=\{1,2\}$ and $s(H)\in\{\SG(\varepsilon),\SG(\ast)\}$. Applications of the rule $\pi$ to $\SG(\ast)$ yield exactly the nonempty unlabeled string graphs. Therefore, the latter is equivalent to $H$ being a string graph. This proves that being a string graph is \EXPTIME{}-compatible.

	In order to prove that non-being a string graph is \EXPTIME{}-compatible, it suffices to redefine
	\[
	I_0(X) \eqdef
	\bigl(\{G\in\mathcal H(\{\ast\})\mid G\text{ is good},\ \type(G)=X\}
	\setminus\{\SG(\varepsilon),\SG(\ast)\}\bigr)
	\cup\{\mathit{bad}\}. \qedhere
	\]
\end{proof}

Since not being a string graph is both non-Parikh and \EXPTIME{}-compatible, we get the following result.
\begin{corollary}\label{corollary:string-generating}
	The problem whether a given HRG is string-generating is \EXPTIME{}-complete.
\end{corollary}

\section{Discussion and Conclusion}\label{section:conclusion}

What are the benefits of having studied complexity of the uniform membership problem for HRGs? First, our results show that the choice of the definition of a hypergraph, which varies in different papers, might affect complexity of this problem significantly. It is a consequence of the paper's results that converting a string-generating HRG into an equivalent repetition-free one cannot be done in polynomial time, assuming that $\NP{}\ne\EXPTIME{}$. 

Secondly, the complexity results can be used for other kinds of graph grammars. One particular example is fusion grammars, introduced in \cite{KreowskiKL17}. They extend hyperedge replacement grammars by the operation of fusing hyperedges within a hypergraph; these grammars are motivated, in particular, by interactions of DNA molecules. In \cite{Pshenitsyn25}, it is proved that the uniform membership problem for fusion grammars without markers and connectors is decidable and lies in NEXPTIME. Since an HRG can be transformed into a fusion grammar in polynomial time, the following is implied by \cref{theorem:lower_bound_EXPTIME}.
\begin{corollary}\label{corollary:fusion-grammar}
	The uniform membership problem for fusion grammars is EXPTIME-hard.
\end{corollary}

It is essential that hypergraphs in fusion grammars are repetition-allowing, because a fusion rule application, starting even with a repetition-free hypergraph, can transform it into a hypergraph with repetitions of attachment nodes. Thus, one should be aware that it is hard to avoid \EXPTIME{}-hardness of the uniform membership problem for fusion grammars. (Note that little is known about non-uniform membership for those.)

Thirdly, the results of the paper increase one's awareness of the difference between non-uniform and uniform membership problems for HRGs. For repetition-free ones, both problems are \NP{}-complete but, for repetition-allowing ones, we again have \NP{} versus \EXPTIME{}. 

Fourthly, and probably most importantly, the methods used in the \EXPTIME{} lower bound turned out to be very prolific, applicable to a very wide range of decision problems concerning HRGs. \cref{theorem:main-meta,theorem:main-meta-repetition-free} are nice generalisations I did not expect to discover initially.

Finally, the complexity results of the paper incorporate HRGs in the general complexity picture of mildly context-sensitive grammar formalisms. Although HRGs are closer to LCFRSs than to MCFGs because the former two are both information-lossless, unlike the latter, it turns out that there is still a way to exploit flexibility of hypergraphs to encode \EXPTIME{}-complete problems in HRGs, even if they generate only string graphs.

\section*{Acknowledgments}

I am grateful to Frank Drewes for his attention to this work and for fruitful discussions, in particular, for his suggestion to look at automata over permutation groups. I also thank the anonymous reviewers for the valuable remarks that helped me to improve the presentation.

\section*{Funding}

This work was funded by the Ministry of Science and Higher Education of the Russian Federation (Grant No. 075-15-2024-529).

\section*{Disclosure on Applying AI Tools}

The results of this work have been obtained without AI assistance, except \cref{corollary:colorable}, which was proposed by GPT-6 Astra. GPT-6.1 Sol and GPT-6 Astra were also used for polishing the text and fixing minor mistakes.





\bibliographystyle{plain}
\bibliography{HRG_Complexity_bib}

\end{document}